\documentclass{article}
\usepackage{geometry}

\usepackage[T1]{fontenc}
\usepackage[utf8]{inputenc}
\usepackage{authblk} 

\usepackage{algorithm}
\usepackage{cite}
\usepackage[noend]{algorithmic}

\usepackage{amsthm}

\newtheorem{definition}{Definition}
\newtheorem{theorem}{Theorem}
\newtheorem{lemma}{Lemma}
\newtheorem{corollary}{Corollary}
\newtheorem{observation}{Observation}

\newtheorem{proposition}{Proposition}

\usepackage{amsmath,amssymb,graphicx}
\usepackage{xcolor}
\usepackage[colorinlistoftodos,prependcaption,textsize=tiny]{todonotes}
\usepackage{subcaption}
\usepackage{graphicx} 
\usepackage[inline]{enumitem} 

\graphicspath{ {./images/} }
\usepackage{tcolorbox}
\usepackage{amsfonts}

\usepackage{mathtools}
\DeclarePairedDelimiter{\ceil}{\lceil}{\rceil}

\usepackage{fancyhdr}

\newcount\Comments  
\definecolor{darkgreen}{rgb}{0,0.6,0}
\newcommand{\kibitz}[2]{\ifnum\Comments=1{\color{#1}{#2}}\fi}

\title{The Complexity of Growing a Graph}

\author[1]{George Mertzios}
\author[2]{Othon Michail}
\author[3]{George Skretas}
\author[2]{Paul G. Spirakis}
\author[2]{Michail Theofilatos}
\affil[1]{Department of Computer Science, Durham University, Durham, UK}
\affil[2]{Department of Computer Science, University of Liverpool, Liverpool, UK}
\affil[3]{Hasso Plattner Institute, University of Potsdam, Potsdam, Germany}

\begin{document}


\maketitle

\begin{abstract}

We study a new algorithmic process of graph growth which starts from a single initial vertex and operates in discrete time-steps, called \emph{slots}. In every slot, the graph grows via two operations (i) vertex generation and (ii) edge activation. The process completes at the last slot where a (possibly empty) subset of the edges of the graph will be removed. Removed edges are called \emph{excess edges}. The main problem investigated in this paper is: Given a target graph $G$, we are asked to design an algorithm that outputs such a process growing $G$, called a \emph{growth schedule}. Additionally, the algorithm should try to minimize the total number of slots $k$ and of excess edges $\ell$ used by the process. We provide both positive and negative results for different values of $k$ and $\ell$, with our main focus being either schedules with sub-linear number of slots or with zero excess edges.
\end{abstract}


\section{Introduction}
\label{sec:introduction}

\subsection{Motivation}
\label{subsec:motivation}

Growth processes are found in a variety of networked systems. In nature, crystals grow from an initial nucleation or from a ``seed'' crystal and a process known as embryogenesis develops sophisticated multicellular organisms, by having the genetic code control tissue growth \cite{chan2019molecular,Hernandez2019signaling}. In human-made systems, sensor networks are being deployed incrementally to monitor a given geographic area \cite{andrew2002incremental,chatzi2008adaptive}, social-network groups expand by connecting with new individuals \cite{Cordeiro18}, DNA self-assembly automatically grows molecular shapes and patterns starting from a seed assembly \cite{rothemund2006folding,doty2012theory,woods2013active}, and high churn or mobility can cause substantial changes in the size and structure of computer networks \cite{becchetti2021expansion,augustine2012towards}. Graph-growth processes are central in some theories of relativistic physics. For example, in dynamical schemes of causal set theory, causets develop from an initial emptiness via a tree-like birth process, represented by dynamic Hasse diagrams \cite{bombelli1987space,rideout1999classical}.

Though diverse in nature, all these are examples of systems sharing the notion of an underlying graph-growth process. In some, like crystal formation, tissue growth, and sensor deployment, the implicit graph representation is bounded-degree and embedded in Euclidean geometry. In others, like social-networks and causal set theory, the underlying graph might be free from strong geometric constraints but still be subject to other structural properties, as is the special structure of causal relationships between events in casual set theory.

Further classification comes in terms of the source and control of the network dynamics. Sometimes, the dynamics are solely due to the environment in which a system is operating, as is the case in DNA self-assembly, where a pattern grows via random encounters with free molecules in a solution. In other applications, the network dynamics are, instead, governed by the system. Such dynamics might be determined and controlled by a centralized program or schedule, as is typically done in sensor deployment, or be the result of local independent decisions of the individual entities of the system. Even in the latter case, the entities are often running the same global program, as do the cells of an organism by possessing and translating the same genetic code.

Inspired by such systems, we study a high-level, graph-theoretic abstraction of network-growth processes. We do not impose any strong \emph{a priori} constraints, like geometry, on the graph structure and restrict our attention to \emph{centralized} algorithmic control of the graph dynamics. We do include, however, some weak conditions on the permissible dynamics, necessary for non-triviality of the model and in order to capture realistic abstract dynamics. One such condition is ``locality'', according to which a newly introduced vertex $u^\prime$ in the neighborhood of a vertex $u$, can only be connected to vertices within a reasonable distance $d-1$ from $u$. 
At the same time, we are interested in growth processes that are ``efficient'', under meaningful abstract measures of efficiency. We consider two such measures, to be formally defined later, the \emph{time} to grow a given target graph and the number of auxiliary connections, called \emph{excess edges}, employed to assist the growth process. For example, in cellular growth, a useful notion of time is the number of times all existing cells have divided and is usually polylogarithmic in the size of the target tissue or organism. In social networks, it is quite typical that new connections can only be revealed to an individual $u^\prime$ through its connection to another individual $u$ who is already a member of a group. Later,  $u^\prime$ can drop its connection to $u$ but still maintain some of its connections to $u$'s group. The dropped connection $uu^\prime$ can be viewed as an excess edge, whose creation and removal has an associated cost, but was nevertheless necessary for the formation of the eventual neighborhood of $u^\prime$.     

The present study is also motivated by recent work on dynamic graph and network models \cite{michail2018elements,michail2016introduction,casteigts2012time}. 
Research on temporal graphs studies the algorithmic and structural properties of graphs $\mathcal{G}=(\mathcal{V},\mathcal{E})$, in which $\mathcal{V}$ is a set of time-vertices and $\mathcal{E}$ a set of time-edges of the form $(u,t)$ and $(e,t)$, respectively, $t$ indicating the discrete time at which an instance of vertex $u$ or edge $e$ is available. A substantial part of work in this area has focused on the special case of temporal graphs in which $\mathcal{V}$ is \emph{static}, i.e.,~time-invariant \cite{kempe2002connectivity,berman1996vulnerability,mertzios2019temporal, enright2019deleting, zschoche2020complexity,AkridaMSZ20}.
In overlay networks \cite{AACW05,aspnes2007skip,gotte2019faster,Gilbert2020Dconstructor,thorsten2021time} and distributed network reconfiguration \cite{michail2022distributed}, $\mathcal{V}$ is a static set of processors that control in a decentralized way the edge dynamics. Even though, in this paper, we do not study our dynamic process from a distributed perspective, it still shares with those models both the fact that dynamics are \emph{active}, i.e.,~algorithmically controlled, and the locality constraint on the creation of new connections.
Nevertheless, our main motivation is theoretical interest. As will become evident, the algorithmic and structural properties of the considered graph-growth process give rise to some intriguing theoretical questions and computationally hard combinatorial optimization problems. Apart from the aforementioned connections to dynamic graph and network models, we shall reveal interesting similarities to cop-win graphs \cite{LinSS12,Bandelt91,poston71,chepoi98}. It should also be mentioned that there are other well-studied models and processes of graph growth, not obviously related to the ones considered here, such as random graph generators \cite{bollobas2001,KKPR99}. While initiating this study from a centralized and abstract viewpoint, we anticipate that it can inspire work on more applied models, including geometric ones and models in which the growth process is controlled in a decentralized way by the individual network processors.  Note that centralized upper bounds can be translated into (possibly inefficient) first distributed solutions, while lower bounds readily apply to the distributed case. There are other recent studies considering the centralized complexity of problems with natural distributed analogues, as is  the work of Scheideler and Setzer on the centralized complexity of transformations for overlay networks \cite{scheideler2019complexity} and of some of the authors of this paper on geometric transformations for programmable matter \cite{michail2019transformation}.

\subsection{Our Approach}
\label{sec:approach}

We study the following centralized graph-growth process. The process, starting from a single initial vertex $u_0$ and applying vertex-generation and edge-modification operations, grows a given target graph $G$. It operates in discrete time-steps, called slots. In every slot, it generates at most one new vertex $u^\prime$ for each existing vertex $u$ and connects it to $u$. This is an operation abstractly representing entities that can replicate themselves or that can attract new entities in their local neighborhood or group. Then, for each new vertex $u^\prime$, it connects $u^\prime$ to any (possibly empty) subset of the vertices within a ``local'' radius around $u$, described by a distance parameter $d$, essentially representing that radius plus 1, i.e., as measured from $u^\prime$. Finally, it removes any (possibly empty) subset of edges whose removal does not disconnect the graph, before moving on to the next slot. These edge-modification operations are essentially capturing, at a high level, the local dynamics present in most of the applications discussed previously. In these applications, new entities typically join a local neighborhood or a group of other entities, which then allows them to easily connect to any of the local entities. Moreover, in most of these systems, existing connections can be easily dropped by a local decision of the two endpoints of that connection. \footnote{Despite locality of new connections, a more global effect is still possible. One is for the degree of a vertex $u$ to be unbounded (e.g.,~grow with the number of vertices). Then $u^\prime$, upon being generated, can connect to an unbounded number of vertices within the ``local'' radius of $u$. Another would be to allow the creation of connections between vertices generated in the past, which would enable local neighborhoods to gradually grow unbounded through transitivity relations. In this work, we do allow the former but not the latter. That is, for any edge $(u,u')$ generated in slot $t$, it must hold that $u$ was generated in some slot $t_{past}<t$ while $u'$ was generated in slot $t$. Other types of edge dynamics are left for future work.} The rest of this paper exclusively focuses on $d=2$. 
It is not hard to observe that, without additional considerations, any target graph can be grown by the following straightforward process. In every slot $t$, the process generates a new vertex $u_t$ which it connects to $u_0$ and to all neighbors of $u_0$. The graph grown by this process by the end of slot $t$, is the clique $K_{t+1}$, thus, any $K_n$ is grown by it within $n-1$ slots. As a consequence, any target graph $G$ on $n$ vertices can be grown by extending the above process to first grow $K_n$ and then delete all edges in $E(K_n)\setminus E(G)$, at the end of the last slot. Such a clique growth process maximizes both complexity parameters that are to be minimized by the developed processes. One is the time to grow a target graph $G$, to be defined as the number of slots used by the process to grow $G$, and the other is the total number of deleted edges during the process, called excess edges. The above process always uses $n-1$ slots and may delete up to $\Theta(n^2)$ edges for sparse graphs, such as a path  graph or a planar graph. 

There is an improvement of the clique process, which connects every new vertex $u_t$ to $u_0$ and to exactly those neighbors $v$ of $u_0$ for which $vu_t$ is an edge of the target graph $G$. At the end, the process deletes those edges incident to $u_0$ that do not correspond to edges in $G$, in order to obtain $G$. If $u_0$ is chosen to represent the maximum degree, $\Delta(G)$, vertex of $G$, then it is not hard to see that this process uses $n-1-\Delta(G)$ excess edges, while the number of slots remains $n-1$ as in the clique process. However, we shall show that there are (poly)logarithmic-time processes using close to linear excess edges for some of those graphs. In general, processes considered \emph{efficient} in this work will be those using (poly)logarithmic slots and linear (or close to linear) excess edges.

The goal of this paper is to investigate the algorithmic and structural properties of such processes of graph growth, with the main focus being on studying the following combinatorial optimization problem, which we call the \emph{Graph Growth Problem}. In this problem, a centralized algorithm is provided with a target graph $G$, usually from a graph family $F$, and non-negative integers $k$ and $\ell$ as its input. The goal is for the algorithm to compute, in the form of a \emph{growth schedule} for $G$, such a process growing $G$ with at most $k$ slots and using at most $\ell$ excess edges, if one exists. All algorithms we consider are polynomial-time.\footnote{Note that this reference to \emph{time} is about the running time of an algorithm computing a growth schedule. But the length of the growth schedule is another representation of time: the time required by the respective growth process to grow a graph. To distinguish between the two notions of time, we will almost exclusively use the term \emph{number of slots} to refer to the length of the growth schedule and \emph{time} to refer to the running time of an algorithm generating the schedule.}

For an illustration of the discussion so far, consider the graph family $F_{star}=\{G \;|\; G$ is a star on $n=2^\delta$ vertices$\}$ and assume that edges are activated within local distance $d=2$. We describe a simple algorithm returning a time-optimal and linear excess-edges growth process, for any target graph $G\in F_{star}$ given as input. To keep this exposition simple, we do not give $k$ and $\ell$ as input-parameters to the algorithm. The process computed by the algorithm, shall always start from $G_0=(\{u_0\},\emptyset)$. In every slot $t=1,2,\ldots,\delta$ and every vertex $u\in V(G_t)$ the process generates a new vertex $u^\prime$, which it connects to $u$. If $t>1$ and $u\neq u_0$, it then activates the edge $u_0u^\prime$, which is at distance 2, and removes the edge $uu^\prime$. It is easy to see that by the end of slot $t$, the graph grown by this process is a star on $2^t$ vertices centered at $u_0$, see Figure \ref{fig:star}. Thus, the process grows the target star graph $G$ within $\delta=\log n$ slots. By observing that $2^t/2 -1$ edges are removed in every slot $t$, it follows that a total of $\sum_{1\leq t\leq \log n} 2^{t-1} -1 < \sum_{1\leq t\leq \log n} 2^t = O(n)$ excess edges are used by the process. Note that this algorithm can be easily designed to compute and return the above growth schedule for any $G\in F_{star}$ in time polynomial in the size $|\langle G\rangle|$ of any reasonable representation of $G$.

\begin{figure}[!hbtp]
   \centering{
        \begin{subfigure}{0.20\textwidth}
        \includegraphics[width=\textwidth]{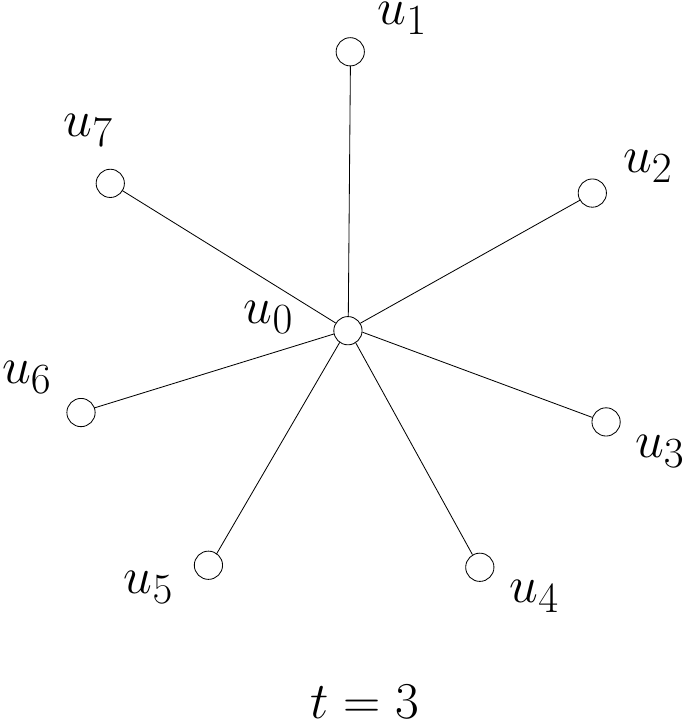}
        \caption{}
        \label{fig:star-1}
        \end{subfigure}
	\hspace{0.6cm}
        \begin{subfigure}{0.20\textwidth}
        \includegraphics[width=\textwidth]{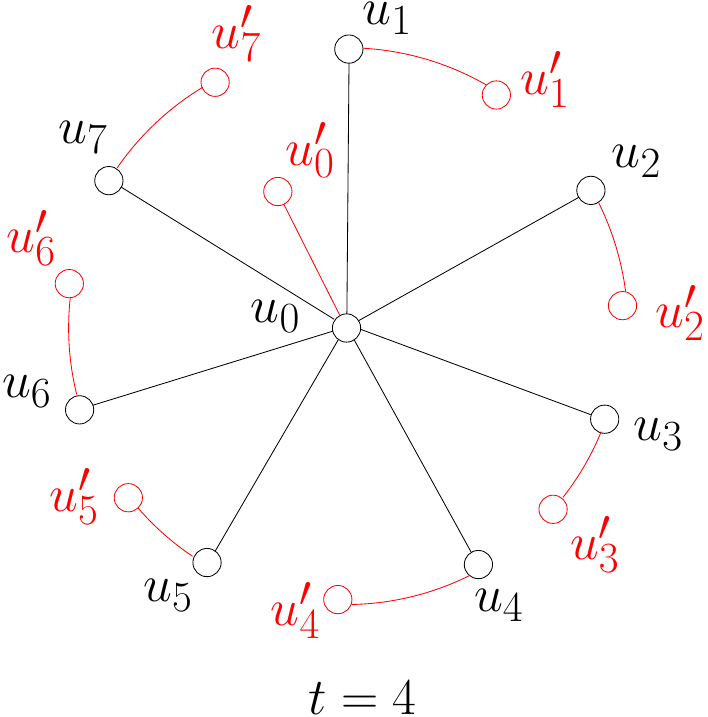}
        \caption{}
        \label{fig:star-2}
        \end{subfigure}
	\hspace{0.6cm}
        \begin{subfigure}{0.20\textwidth}
        \includegraphics[width=\textwidth]{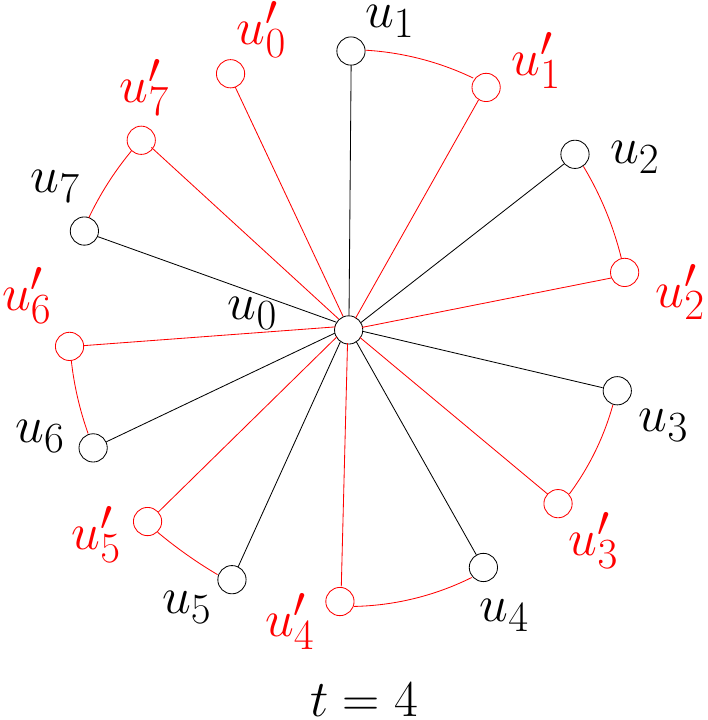}
        \caption{}
        \label{fig:star-3}
        \end{subfigure}
	\hspace{0.6cm}
        \begin{subfigure}{0.20\textwidth}
        \includegraphics[width=\textwidth]{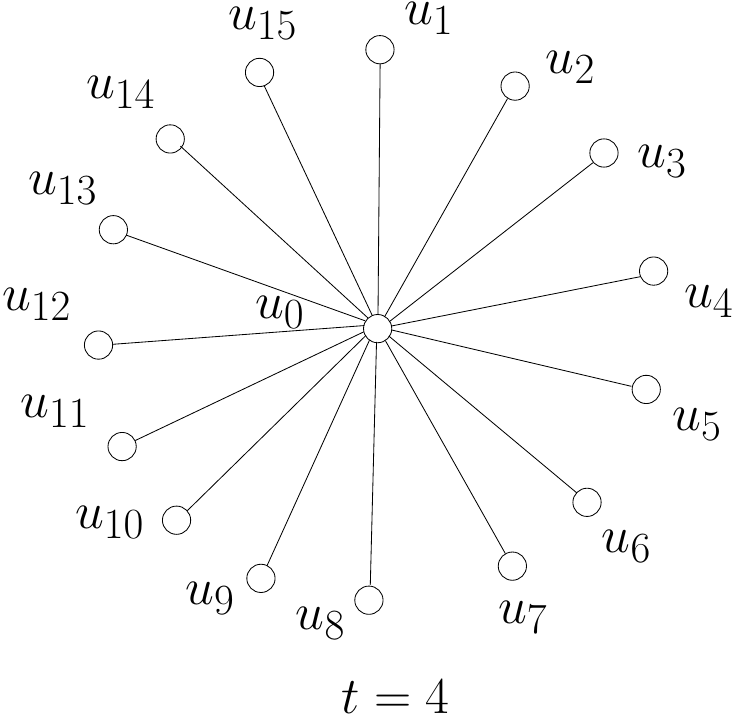}
        \caption{}
        \label{fig:star-4}
        \end{subfigure}
        }
   \caption{The operations of the star graph process in slot $t=4$. (a) A star of size $2^3$ grown by the end of slot 3. (b) For every $u_i$, a vertex $u'_i$ is generated by the process and is connected to $u_i$. (c) New vertices $u'_i$ are connected to $u_0$. (d) Edges between peripheral-vertices are being removed to obtain the star of size $2^4$ grown by the end of slot 4. Here, we also rename the vertices for clarity.} \label{fig:star}
\end{figure}

Note that there is a natural trade-off between the number of slots and the number of excess edges that are required to grow a target graph. That is, if we aim to minimize the number of slots (resp.~of excess edges) then the number of excess edges (resp.~slots) increases. To gain some insight into this trade-off, consider the example of a path graph $G$ on $n$ vertices $u_0,u_1,...,u_{n-1}$, where $n$ is even for simplicity. 
If we are not allowed to activate any excess edges, then the only way to grow $G$ is to always extend the current path from its endpoints, which implies that a schedule that grows $G$ must have at least $\frac{n}{2}$ slots.  
Conversely, if the growth schedule has to finish after $\log n$ slots, then $G$ can only be grown by activating $\Omega(n)$ excess edges.

In this paper, we mainly focus on this trade-off between the number of slots and the number of excess edges that are needed to grow a specific target graph $G$.
In general, given a growth schedule $\sigma$, any excess edge can be removed just after the last time it is used as a ``relay'' for the activation of another edge. In light of this, an algorithm computing a growth schedule can spend linear additional time to optimize the slots at which excess edges are being removed. 
A complexity measure capturing this is the \emph{maximum excess edges lifetime}, defined as the maximum number of slots for which an excess edge remains active. Our algorithms will generally be aiming to minimize this measure. 
When the focus is more on the trade-off between the slots and the number of excess edges, we might be assuming that all excess edges are being removed in the last slot of the schedule, as the exact timing of deletion makes no difference w.r.t. these two measures.

\subsection{Contribution}\label{sec:contribution}

Section 2 begins by presenting the model and problem statement. We continue with Section \ref{sec:subcases} where we investigate cases for edge activation distance $d=1$ and $d\geq3$. We finish the section by listing some basic properties and two basic sub-processes that are recurrent in our growth processes. 
In Section \ref{sec:edge_act_distance_2}, we provide some basic propositions that are crucial to understanding the limitations on the number of slots and the number of excess edges required for a growth schedule of a graph~$G$. We then use these propositions to provide some lower bounds on the number of slots. 

In Section \ref{sec:zero_excess}, we study the \emph{zero-excess growth schedule} problem, where the goal is to decide whether a graph $G$ has a growth schedule of $k$ slots and $\ell=0$ excess edges. We define the {\em candidate elimination ordering} of a graph $G$ as an ordering $v_1,v_2, \ldots,v_n$ of $V(G)$ so that for every vertex $v_i$, there is some $v_j$, where $j<i$ such that $N[v_i]\subseteq N[v_j]$ in the subgraph induced by $v_i, \ldots,v_n$, for $1\leq i\leq n$. We show that a graph has a growth schedule of $k = n - 1$ slots and $\ell = 0$ excess edges if and only if it is has a {\em candidate elimination ordering}. Our main positive result is a polynomial-time algorithm that computes whether a graph has a growth schedule of $k=\log n$ slots and $\ell = 0$ excess edges. If it does, the algorithm also outputs such a growth schedule. On the negative side, we give two strong hardness results. We first show that the decision version of the zero-excess growth schedule problem is NP-complete. Then, we prove that, for every $\varepsilon>0$, there is no polynomial-time algorithm which computes a $n^{\frac{1}{3} - \varepsilon}$-approximate zero-excess growth schedule, unless P~=~NP.

In Section \ref{sec:algorithms_basic_graph_classes}, we study growth schedules of (poly)logarithmic slots. We provide two polynomial-time algorithms. One outputs, for any tree graph, a growth schedule of $O(\log^2n)$ slots and only $O(n)$ excess edges, and the other outputs, for any planar graph, a growth schedule of $O(\log n)$ slots and $O(n\log n)$ excess edges. Finally, we give lower bounds on the number of excess edges required to grow a graph, when the number of slots is fixed to $\log n$.

In Section \ref{sec:Conclusion} we discuss some interesting problems opened by this work.

\section{Preliminaries}
\label{sec:preliminaries}

\subsection{Model and Problem Statement}\label{subsec:model} 

A {\em growing graph} is modeled as an undirected dynamic graph $G_t = (V_t, E_t)$, where $t=1,2,\ldots,k$ is a discrete time-step, called {\em slot}. 
The dynamics of $G_t$ are determined by a centralized \emph{growth process} (also called \emph{growth schedule}) $\sigma$, defined as follows. The process always starts from the initial graph instance $G_0=(\{u_0\},\emptyset)$, containing a single initial vertex $u_0$, called the \emph{initiator}. 
In every slot $t$, the process updates the current graph instance $G_{t-1}$ to generate the next, $G_t$, according to the following vertex and edge update rules. The process first sets $G_t = G_{t-1}$. Then, for every $u\in V_{t-1}$, it adds at most one new vertex $u^\prime$ to $V_t$ (\emph{vertex generation} operation) and adds to $E_t$ the edge $uu^\prime$ alongside any subset of the edges $\{vu^\prime\;|\; v\in V_{t-1}$ is at distance at most $d-1$ from $u$ in  $G_{t-1}\}$, for some integer \emph{edge-activation distance} $d\geq 1$ fixed in advance (\emph{edge activation} operation).  We call $u^\prime$ the vertex generated by the process for vertex $u$ in slot $t$. We also say that $u$ is the \emph{parent} of $u^\prime$ and that $u^\prime$ is the \emph{child} of $u$ at slot $t$ and write $u \overset{t}{\rightarrow} u^\prime$. The process completes slot $t$ after deleting any (possibly empty) subset of edges from $E_t$ (\emph{edge deletion} operation). We also denote by $V^+_t$, $E^+_t$, and $E^-_t$ the set of vertices generated, edges activated, and edges deleted in slot $t$, respectively. Then, $G_{t} = (V_{t}, E_{t})$ is also given by $V_{t} = V_{t-1} \cup V^+_t$ and $E_{t} = (E_{t-1} \cup E^+_t) \setminus E^-_t$. Deleted edges are called \emph{excess edges} and we restrict attention to excess edges whose deletion does not disconnect $G_t$. We call $G_t$ the graph \emph{grown} by process $\sigma$ after $t$ slots and call the final instance, $G_k$, the \emph{target graph} grown by $\sigma$. We also say that $\sigma$ is a \emph{growth schedule for} $G_k$ that grows $G_k$ in $k$ slots using $\ell$ {\em excess edges}, where $\ell$=$\sum_{t=1}^{k}|E^-_t|$, i.e., $\ell$ is equal to the total number of deleted edges. 
This brings us to the main problem studied in this paper:\\

\vspace{-5pt}
\noindent\textbf{Graph Growth Problem:} Given a target graph $G$ and non-negative integers $k$ and $\ell$, compute a growth schedule for $G$ of at most $k$ slots and at most $\ell$ excess edges, if one exists.\\

\vspace{-5pt}
The \emph{target graph} $G$, which is part of the input, will often be drawn from a given graph family $F$, e.g.,~the family of planar graphs. Throughout, $n$ denotes the number of vertices of the target graph $G$. In this paper, computation is always to be performed by a \emph{centralized} polynomial-time algorithm. 

Let $w$ be a vertex generated in a slot $t$, for $1\leq t \leq k$. The \emph{birth path} of vertex $w$ is the unique sequence $B_w=(u_0, u_{i_1}, \dots, u_{i_{p-1}}, u_{i_p}=w)$ of vertices, where $i_p = t$ and $u_{i_{j-1}}\overset{i_j}{\rightarrow}u_{i_j}$, for every $j=1,2,\ldots,p$. That is, $B_w$ is the sequence of vertex generations that led to the generation of vertex~$w$. Furthermore, the \emph{progeny} of a vertex $u$ is the set $P_u$ of descendants of $u$, i.e.,~$P_u$ contains those vertices $v$ for which $u\in B_v$ holds. We also define the sets $N(u)$ and $N[u]$ to be the neighborhood of $u$ and closed neighborhood of $u$, respectively.

In what follows, we give a formal and detailed description of a graph growth schedule. We will be using this description for the pseudocode of our algorithms.

\begin{definition}[growth schedule for $d=2$]\label{construction-schedule-def}
Let $\sigma = (\mathcal{S}_1, \mathcal{S}_2, \ldots, \mathcal{S}_k, \mathcal{E})$ be an ordered sequence of sets, 
where $\mathcal{E}$ is a set of edges, and 
each $\mathcal{S}_i=\{(u_1,v_1,E_1)$, $(u_2,v_2,E_2),\ldots,(u_q,v_q,E_q)\}$ is an unordered set of ordered tuples $\{(u_j,v_j,E_j):1\leq j\leq q\}$ such that, for every $j$, $u_j$ and $v_j$ are vertices (where $u_j$ gives birth to $v_j$) and $E_j$ is a set of edges incident to $v_j$ such that $u_jv_j\in E_j$. Suppose that, for every $2\leq i\leq k$, the following conditions are all satisfied:

\begin{itemize}
    \item each of the sets $\{v_1,v_2,\ldots,v_q\}$ and $\{u_1,u_2,\ldots,u_q\}$ contain $q$ distinct vertices,
    
    \item each vertex $v_j\in\{v_1,v_2,\ldots,v_q\}$ does not appear in any set among $\mathcal{S}_1,\ldots,\mathcal{S}_{i-1}$ 
    (i.e.,~$v_j$ is ``born'' at slot $i$),
    
    \item for each vertex $u_j\in\{u_1,u_2,\ldots,u_q\}$, there exists exactly one set among $\mathcal{S}_1,\ldots,\mathcal{S}_{i-1}$ which contains a tuple $(u',u_j,E')$ 
    (i.e.,~$u_j$ was ``born'' at a slot before slot $i$).
\end{itemize}

Let $i\in\{2,\ldots,k\}$, and let $u$ be a vertex that has been generated at some \emph{slot} 
$i'\leq i$, that is, $u$ appears in at least one tuple of a set among $\mathcal{S}_1,\ldots,\mathcal{S}_{i}$. 
We denote by~$E^{i}$ the union of all edge sets that appear in the tuples of the sets $\mathcal{S}_1,\ldots,\mathcal{S}_{i}$; $E^i$ is the set of all edges activated until slot $i$. 
We denote by $N_i(u)$ the set of neighbors of $u$ in the set~$E^i$.
If, in addition, $\mathcal{E}\subseteq E^k$ and, 
for every $2\leq i\leq k$ and for every tuple $(u_j,v_j,E_j)\in\mathcal{S}_i$, we have that $N_i(v_j)\subseteq N_i(u_j)$, 
then $\sigma$ is a \emph{growth schedule} for the graph $G=(V,E^k\setminus\mathcal{E})$, 
where $V$ is the set of all vertices which appear in at least one tuple in $\sigma$. 
The number $k$ of sets in $\sigma$ is the \emph{length} of $\sigma$. 
Finally, we say that $G$ is \emph{constructed} by $\sigma$ with $k$ slots and with $|\mathcal{E}|$ excess edges.
\end{definition}

\subsection{The Case $d=1$ and $d \geq 3$}\label{sec:subcases}

In this section, we show that for edge-activation distance $d=1$ or $d \geq 3$ there are simple but very efficient algorithms for finding growth schedules. 
As a warm-up, we begin with a simple observation for the special case where the edge-activation distance $d$ is equal to 1.

\begin{observation}
    For $d=1$, every graph $G$ that has a growth schedule is a tree graph.
\end{observation}

\begin{proposition}\label{prop:line}
    For $d=1$, the shortest growth schedule $\sigma$ of a path graph (resp.~a star graph) on $n$ vertices has $\ceil{n/2}$ (resp.~$n-1$) slots.
\end{proposition}

\begin{proof}
    Let $G$ be the path graph on $n$ vertices. By definition of the model for $d=1$, edges can only be activated during vertex generation, between the generated vertex and its father. Thus, increasing the size of the path can only be achieved by generating one new vertex at each of the endpoints of the path. The size of a path can only be increased by at most $2$ in each slot, where for each endpoint of the path a new vertex that becomes the new endpoint of the path is generated. Therefore, in order to create any path graph of size $n$ would require at least $\ceil{n/2}$ slots. The growth schedule where one vertex is generated at each of the endpoints of the path in each slot creates the path graph of $n$ vertices in $\ceil{n/2}$ slots.
    
    Now let $G$ be the star graph with $n-1$ leaves. Increasing the size of the star graph can only be achieved by generating new leaves directly connected to the center vertex, and this can occur at most once per slot. Therefore, the growth schedule of $G$ requires exactly $n-1$ slots.
\end{proof}

\begin{observation}
Let $d=1$ and $G=(V,E)$ be a tree graph with diameter $diam$. Then any growth schedule $\sigma$ that grows $G$ requires at least $\ceil{diam/2}$ slots.
\end{observation}

\begin{proof}
Consider a path $p$ of size $diam$ that realizes the diameter of graph $G$. By Proposition \ref{prop:line} we know that $p$ alone requires a growth schedule with at least $\ceil{diam/2}$ slots.
\end{proof}

\begin{observation}\label{leaf}
Let $d=1$ and $G=(V,E)$ be a tree graph with maximum degree $\Delta(G)$. Then any growth schedule $\sigma$ that grows $G$ requires at least $\Delta(G)$ slots.
\end{observation}

\begin{proof}
    Consider a vertex $u \in G$ with degree $\Delta(G)$ and let $G'=(V',E')$ be a subgraph of $G$, such that $V' = N[u]$ (the close neighborhood of vertex $u$) and $E' = E(N[u])$.
    Notice that $G'$ is a star graph of size $\Delta(G)$. By Proposition \ref{prop:line}, we know that the growth schedule of $G'$ alone has at least $\Delta(G)$ slots. 
\end{proof}

We now provide an algorithm, called {\em \texttt{trimming}}, that optimally solves the graph growth problem for $d=1$. We begin with the following simple observation.

\begin{observation}
    Let $d=1$. Consider a tree graph $G$ and a growth schedule $\sigma$ for it. Denote by $G_t$ the graph grown so far until the end of slot $t$ of $\sigma$. Then any vertex generated in slot $t$ must be a leaf vertex in $G_t$.
\end{observation}

\begin{proof}
    Every vertex $u$ generated in slot $t$ has degree equal to 1 at the end of slot $t$ by definition of the model for $d=1$. Therefore every vertex $u$ in graph $G_t$ must be a leaf vertex.
\end{proof}

The {\em \texttt{trimming}} algorithm, see Algorithm \ref{Trimming} follows a bottom-up approach for building the intended growth schedule $\sigma=(\mathcal{S}_k,\mathcal{S}_{k-1},\ldots \mathcal{S}_1,\mathcal{E})$, where $\mathcal{E}=\emptyset$. At every iteration $i=1,2,\ldots$ of the algorithm, we
consider the leaves of the current tree graph and we create the parent-child pairs of the currently last slot $\mathcal{S}_{k+1-i}$ of the schedule. Then we remove from the current tree graph all the leaves that were included in a parent-child pair at this iteration of the algorithm, and we recurse. The process is repeated until graph $G$ has a single vertex left which is added in the first slot of $\sigma$ as the initiator. 
In the next theorem we show that the algorithm produces an optimum growth schedule.

\begin{algorithm}[t!]
\caption{Trimming Algorithm, where $d=1$.}
\label{Trimming}
\begin{algorithmic}[1]
\REQUIRE{A target tree graph $G=(V,E)$ on $n$ vertices.}
\ENSURE{An optimal growth schedule for $G$.}

\STATE{$k\leftarrow 1$}

\WHILE{$V\neq \emptyset$}
    \FOR{each leaf vertex $v\in V$ and its unique neighbor $u\in V$}
        \IF{$u$ is not marked as a ``parent in $S_k$''}
            \STATE{Mark $u$ as a ``parent in $\mathcal{S}_k$''}
	        \STATE{$\mathcal{S}_k \leftarrow \mathcal{S}_k \cup\{(u,v,\{uv\}$)\}}
	        \STATE{$V\leftarrow V\setminus \{v\}$}
	    \ENDIF
	\ENDFOR
	\STATE{$\mathcal{S}_{k+1}=\emptyset$; $k\leftarrow k+1$} 
\ENDWHILE
\RETURN{$\sigma=(\mathcal{S}_k,\mathcal{S}_{k-1},\ldots \mathcal{S}_1,\emptyset)$}
\end{algorithmic}
\end{algorithm}

\begin{theorem}\label{correctness-trimming-thm}
     For $d=1$, the {\em \texttt{trimming}} algorithm computes in polynomial time an optimum (shortest) growth schedule $\sigma$ of $\kappa$ slots for any tree graph $G$.
\end{theorem}

\begin{proof}
Let $\sigma=(\mathcal{S}_1,\ldots, \mathcal{S}_k,\emptyset)$ be the growth schedule obtained by the {\em \texttt{trimming}} algorithm (with input $G$). Suppose that $\sigma$ is not optimum, and let $\sigma'\neq \sigma$ be an optimum growth schedule for $G$. That is, $\sigma=(\mathcal{S}'_1,\ldots, \mathcal{S}'_{k'},\emptyset)$, where $k' < k$. 
Denote by $(L_1,L_2,\ldots L_k)$ and $(L'_1,L'_2,\ldots L'_{k'})$ the sets of vertices generated in each slot of the growth schedules $\sigma$ and $\sigma'$, respectively. Note that $\sum_{i=1}^{k}|L_i| = \sum_{i=1}^{k'}|L'_i|=n-1$. 
Among all optimum growth schedules for $G$, we can assume without loss of generality that $\sigma'$ is chosen such that 
the vector $(|L'_{k'}|,|L'_{k'-1}|,\ldots |L'_{1}|)$ is lexicographically largest. 

Recall that the {\em \texttt{trimming}} algorithm builds the growth schedule $\sigma$ backwards, i.e.,~it first computes $\mathcal{S}_k$, it then computes $\mathcal{S}_{k-1}$ etc. At the first iteration, the {\em \texttt{trimming}} algorithm collects all vertices which are parents of at least one leaf in the tree and for each of them will generate a new vertex in $\mathcal{S}_k$. Then, the algorithm removes all leaves that are generated in $\mathcal{S}_k$, and it recursively proceeds with the remaining tree after removing these leaves. 

Let $\ell$ be the number of slots such that the growth schedules $\sigma$ and $\sigma'$ generate the same number of leaves in their last $\ell$ slots, i.e.,~$|L_{k-i}|=|L'_{k'-1}|$, for every $i\in\{0,1,\ldots,\ell-1\}$, but $|L_{k-\ell}|\neq |L'_{k'-\ell}|$. Suppose that $\ell \leq k-1$.
Note by construction of the {\em \texttt{trimming}} algorithm that, since $|L_{k}| = |L'_{k'}|$, both growth schedules $\sigma$ and $\sigma'$ generate exactly one leaf for each vertex which is a parent of a leaf in $G$. That is, in their last slot, both $\sigma$ and $\sigma'$ have the same parents of new vertices; they might only differ in which leaves are generated for these parents. Consider now the graph $G_{k-1}$ (resp.~$G'_{k'-1}$) that is obtained by removing from $G$ the leafs of $L_{k}$ (resp.~of $L'_{k'}$). Then note that $G_{k-1}$ and $G'_{k'-1}$ are isomorphic. Similarly it follows that, if we proceed removing from the current graph the vertices generated in the last $\ell$ slots of the schedules $\sigma$ and $\sigma'$, we end up with two isomorphic graphs $G_{k-\ell+1}$ and $G'_{k'-\ell+1}$. Recall now that, by our assumption, $|L_{k-\ell}|\neq |L'_{k'-\ell}|$. Therefore, since the {\em \texttt{trimming}} algorithm always considers all possible vertices in the current graph which are parents of a leaf (to give birth to a leaf in the current graph), it follows that $|L_{k-\ell}| > |L'_{k'-\ell}|$. That is, at this slot the schedule $\sigma'$ misses at least one potential parent $u$ of a leaf $v$ in the current graph $G'_{k'-\ell+1}$. This means that the tuple $(u,v,\{uv\})$ appears at some other slot $\mathcal{S}'_{j}$ of $\sigma'$, where $j<k'-\ell$. Now, we can move this tuple from slot $\mathcal{S}'_{j}$ to slot $\mathcal{S}'_{k'-\ell}$, thus obtaining a lexicographically largest optimum growth schedule than $\sigma'$, which is a contradiction. 

Therefore $\ell \geq k$, and thus $\ell = k$, since $\sum_{i=1}^{k}|L_i| = \sum_{i=1}^{k'}|L'_i|=n-1$. This means that $\sigma$ and $\sigma'$ have the same length. That is, $\sigma$ is an optimum growth schedule.
\end{proof}

We move on to the case of $d\geq 4$, and we show that for any graph $G$, there is a simple algorithm that computes a growth schedule of an optimum number of slots and only linear number of excess edges in relation to the size of the graph.

\begin{lemma}\label{lemma:d_4}
For $d \geq 4$, any given graph $G=(V,E)$ on $n$ vertices can be grown with a growth schedule $\sigma$ of $\ceil{\log{n}}$ slots and $O(n)$ excess edges.
\end{lemma}

\begin{proof}
    Let $G=(V,E)$ be the target graph, and $G_{t} = (V_{t}, E_{t})$ be the grown graph at the end of slot $t$. When the growth schedule generates a vertex $w$, $w$ is matched with an unmatched vertex of the target graph $G$. For any pair of vertices $v, w \in G_{\ceil{\log{n}}}$ that have been matched with a pair of vertices $v_j, w_j \in G$, respectively, if $(v_j, w_j) \in E$, then $(v, w) \in E_{\ceil{\log{n}}}$, and if $(v_j, w_j) \notin E$, then $(v, w) \notin E_{\ceil{\log{n}}}$.
    
    To achieve growth of $G$ in $\ceil{\log{n}}$ slots, for each vertex of $G_t$ the process must generate a new vertex at slot $t + 1$, except possibly for the last slot of the growth schedule.
    To prove the lemma, we show that the growth schedule maintains a star as a spanning subgraph of $G_t$, for any $t \leq \ceil{\log{n}}$, with the initiator $u$ as the center of the star. Trivially, the children of $u$ belong to the star, provided that the edge between them is not deleted until slot $\ceil{\log{n}}$. The children of all leaves of the star are at distance $2$ from $u$, therefore the edge between them and $u$ are activated at the time of their birth.
    
    The above schedule shows that the distance of any two vertices is always less or equal to four. Therefore, for each vertex $w$ that is generated in slot $t$ and is matched to a vertex $w_j \in G$, the process activates the edges with each vertex $u$ that has been generated and matched to vertex $u_j \in G_j$ and $(w_j, u_j) \in E$. Finally, the number of the excess edges that we activate are at most $2n - 1$ (i.e.,~the edges of the star and the edges between parent and child vertices). Any other edge is activated only if it exists in $G$.
\end{proof}

It is not hard to see that the proof of Lemma~\ref{lemma:d_4} can be slightly adapted such that, instead of maintaining a star, we maintain a clique. The only difference is that, in this case, the number of excess edges increases to at most $O(n^2)$ (instead of at most $O(n)$). On the other hand, this method of always maintaining a clique has the benefit that it works for $d=3$, as the next lemma states.

\begin{lemma}
    For $d \geq 3$, any given graph $G=(V,E)$ on $n$ vertices can be grown with a growth schedule $\sigma$ of $\ceil{\log{n}}$ slots and $O(n^2)$ excess edges.
\end{lemma}

\subsection{Basic Properties and Sub-processes for $d=2$}\label{sec:edge_act_distance_2}

For the rest of the paper, we always assume that $d=2$. In this section, we show some basic properties for growing a graph $G$ which restrict the possible growth schedules and we also provide some lower bounds on the number of slots. We will also provide some basic algorithms which will be used as sub-processes in the rest of the paper. In the next proposition we show that the vertices generated in each slot form an independent set in the grown graph, i.e.,~any pair of vertices generated in the same slot cannot have an edge between them in the target graph.

\begin{proposition} \label{prop:independent set}
    The vertices generated in a slot form an independent set in the target graph $G$.
\end{proposition}

\begin{proof}
    Let $G_{t-1}$ be our graph at the beginning of slot $t$. Consider any pair of vertices $u_1,u_2$ that have minimal distance between them, in other words, they are neighbors and $dist=1$. Assume that for vertices $u_1,u_2$  new vertices $v_1,v_2$ are generated in slot $t$, respectively. The distance between vertices $v_1,v_2$ in slot $t$ just after they are generated is $dist=3$ and therefore the process cannot activate an edge between them. Finally, for any other pair of non-neighboring vertices, the distance between their children is $dist > 3$, thus remaining an independent set.
\end{proof}

\begin{proposition} \label{proposition:distance_reduction}
    Consider any growth schedule $\sigma$ for graph $G$. Let $t_1, t_2$, $t_1 \leq t_2$, be the slots in which a pair of vertices $u, w$ is generated, respectively. Let $dist_{t_2}$ be the distance between $u$ and $w$ at the end of slot $t_2$. Then, at the end of any slot $t\geq t_2$, $dist_{t}\geq dist_{t_2}$.
\end{proposition}

\begin{proof}
    Given that the $d=2$, for any vertex that is generated at slot $t$, edges can only be activated with its parent and with the neighbors of its parent.
    
    Let $u$ be a vertex that is generated for a vertex $u'$ at $t_1$, and $w$ be a vertex that is generated for a vertex $w'$ at $t_2$. Let $G_{t_2-1}$ be the graph at the beginning of slot $t_2$, and $P$, $|P| = d$, be the shortest path between $u$ and $w$ in $G_{t_2-1}$.
    We distinguish two cases:
    
    \vspace{0.2cm}
    
    \begin{enumerate}
        \item For vertex $u$ and/or $w$ new vertices that are connected to all neighbors of $u$ and/or $w$ are generated. In this case, the path that contains the new vertices will clearly be larger than $d$.
        \item In this case, for some vertex $p$ in the path between $u$ and $w$ a new vertex $p'$ is generated and all its edges with the neighbors of $p$ are activated. In this case, the path that passes through $p'$ will clearly have the same size as $P$.
    \end{enumerate}

    It is then obvious that no growth schedule starting from $G_{t_2-1}$ can reduce the shortest distance between $u$ and $w$.
\end{proof}

\begin{proposition}\label{prop:birth path}
    Consider $t_1, t_2$, where $t_1 \leq t_2$, to be the slots in which a pair of vertices $u, w$ is generated by a growth schedule $\sigma$ for graph $G$, respectively, and edge $(u,w)$ is not activated at $t_2$. Then any pair of vertices $v,z$ cannot be neighbors in $G$ if $u$ belongs to the birth path of $v$ and $w$ belongs to the birth path of $z$.
\end{proposition}

\begin{proof}
    Given that the edge between vertices $u$ and $w$ is not activated, and by Proposition \ref{proposition:distance_reduction}, the children of $u$ will always have distance at least $2$ from $w$ (i.e.,~edges of these children can only be activated with the vertices that belong to the neighborhood of their parent vertex, and no edge activations can reduce their distance).
    Sequentially, the same holds also for the children of $w$.
    All vertices that belong to the progeny $P_u$ of $u$ (i.e.,~each vertex $z$ such that $u \in B_z$) have to be in distance at least $2$ from $w$, therefore they cannot be neighbors with any vertex in $P_w$.
\end{proof}

We will now provide some lower bounds on the number of slots for any growth schedule $\sigma$ for graph $G$.

\begin{lemma}
Assume that graph $G$ has chromatic number $\chi (G)$. Then any growth schedule $\sigma$ that grows $G$ requires at least $\chi (G)$ slots.
\end{lemma}

\begin{proof}
Assume that there exists a growth schedule $\sigma$ that can grow graph $G$ in $k<\chi(G)$ slots. By Proposition \ref{prop:independent set}, the vertices generated in each slot $t_i$ for $i=1,2,...,k$ must form an independent set in $G$. Therefore, we could color graph $G$ using $k$ colors which contradicts the original statement that $\chi(G)>k$.
\end{proof}

\begin{lemma}
Assume that graph $G$ has clique number $\omega(G)$. Then any growth schedule $\sigma$ for $G$ requires at least $\omega(G)$ slots.
\end{lemma}

\begin{proof}
By Proposition \ref{prop:independent set}, we know that every slot must contain an independent set of the graph and it cannot contain more than one vertex from clique $q$. By the pigeon hole principle, it follows that $\sigma$ must have at least $c$ slots.
\end{proof}

We continue by presenting simple algorithms for two basic growth processes that are recurrent both in our positive and negative results. One is the process of growing any path graph and the other is that of growing any star graph. Both returned growth schedules use a number of slots which is logarithmic and a number of excess edges which is linear in the size of the target graph. Logarithmic being a trivial lower bound on the number of slots required to grow graphs of $n$ vertices, both schedules are optimal w.r.t. their number of slots. As will shall later follow from Corollary \ref{lower-bound-lem-line} in Section \ref{subsec:lower-bounds}, they are also optimal w.r.t. the number of excess edges used for this time-bound.

\medskip

\noindent\underline{\textbf{{\em \texttt{Path}} algorithm:}} Let $u_0$ always be the ``left'' endpoint of the path graph being grown. 
For any target path graph $G$ on $n$ vertices, the algorithm computes a growth schedule for $G$ as follows. For every slot $1\leq t\leq \ceil{\log{n}}$ and every vertex $u_i\in V_{t-1}$, for $0\leq i \leq 2^{t-1}-1$, it generates a new vertex $u^\prime_i$ and connects it to $u_i$. Then, for all $0\leq i \leq 2^{t-1}-2$, it connects $u^\prime_i$ to $u_{i+1}$ and deletes the edge $u_iu_{i+1}$. Finally, it renames the vertices $u_0,u_1,\ldots,u_{2^t-1}$ from left to right, before moving on to the next slot. 

Figure \ref{fig:line} shows an example slot produced by the path algorithm. The pseudo-code of the algorithm can be found in \ref{alg-line}.  Note that the pseudo-code growth schedule of Algorithm \ref{alg-line} reserves every edge deletion operation until the last slot.

\begin{figure}
    \centering
    \begin{subfigure}[b]{1.0\textwidth}
		\includegraphics[width=1\linewidth]{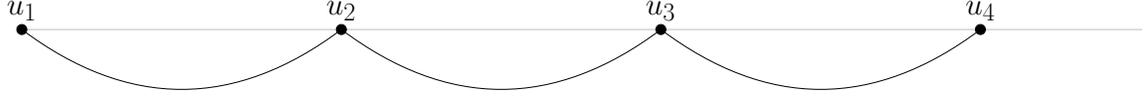}
       \caption{The path graph at the beginning of slot $3$.}
       \label{fig:line_1} 
    \end{subfigure}\hfil \vspace{0.7cm}
    \begin{subfigure}[b]{1.0\textwidth}
		\includegraphics[width=1\linewidth]{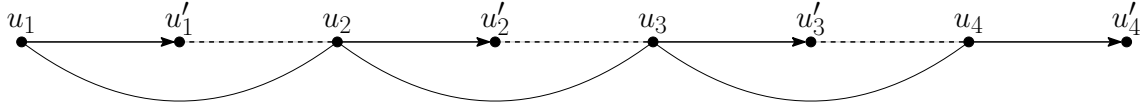}
       \caption{Vertex generation and edge activation step (steps 1 and 2). The arrows represent vertex generations, while dotted lines represent the edges added to vertices of distance 2.}
       \label{fig:line_3}
    \end{subfigure}\hfil \vspace{0.7cm}
    \begin{subfigure}[b]{1.0\textwidth}
		\includegraphics[width=1\linewidth]{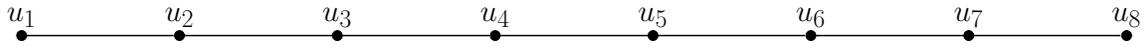}
       \caption{Edge deletion step (step 3) and renaming of vertices.}
       \label{fig:line_4}
    \end{subfigure}
    
    \caption{Third slot of the {\em \texttt{path}} algorithm.}
	\label{fig:line} 
\end{figure}

\begin{lemma}\label{line-construction-lem}
For any path graph $G$ on $n$ vertices, the {\em \texttt{path}} algorithm computes in polynomial time a growth schedule $\sigma$ for $G$ of $\ceil{\log{n}}$ slots and $O(n)$ excess edges.
\end{lemma}

\begin{proof}
    It is easy to see by the description and by Figure \ref{fig:line} that the graph grown is a path subgraph on $n$ vertices. To expand on this, in every slot, we maintain a path graph but we double its size. This process requires $\ceil{\log{n}}$ slots by design since in every slot, for every vertex the process generates a new vertex (apart from the last slot) and after $\ceil{\log{n}}$ slots, the size of the current graph will be $n$. For the excess edges, consider that in the whole process at the end of every slot $t$, every edge activated in the previous slot $t-1$ is deleted. Every edge activated in the process apart from those in the last slot is an excess edge. For every vertex generation there are at most $2$ edge activations that occur in the same slot and there are $n-1$ total vertex generations in total which means that the total edge activations are $2(n-1)$. Therefore, the excess edges are at most $2(n-1)-(n-1)=O(n)$ since the final path graph has $n-1$ edges. Finally, note that if an excess edge is activated in slot $t$, then it is deleted in slot $t+1$ which results in maximum lifetime of $1$.
\end{proof}

\begin{algorithm}[t]
\caption{Path growth schedule.}
\label{alg-line}
\vspace{0.2cm}
\begin{algorithmic}[1]
\REQUIRE{A path graph $G=(V,E)$ on $n$ vertices.}
\ENSURE{A growth schedule for $G$.}

\STATE{$V=u_1$}
\FOR {$k=1,2,\ldots,\ceil{\log n}$}
    \STATE{$\mathcal{S}_k=\emptyset$; \ \ $V_k=\emptyset$; \ \ $\mathcal{E}=\emptyset$}
    \FOR{each vertex $u_i\in V_k$}
        \STATE $\mu(k) = i + \ceil{n/2^{k}}$
        \IF{$u_{\mu(t)} \in V$}
            
            \STATE {$\mathcal{S}_k=\mathcal{S}_k\cup \{(u_i,u_{\mu(k)},\{u_iu_{\mu(k)}, (u_{\mu(k)}u_{\mu(k-1)}: u_{\mu(k-1)} \in V )\})\}$}
            \IF{$k<\ceil{\log n}$}
                \STATE {$\mathcal{E}=\mathcal{E}\cup \{u_iu_{\mu(k)}, (u_{\mu(k)}u_{\mu(k-1)}: u_{\mu(k-1)} \in V )\})\}$}
            \ENDIF    
            \STATE{$V_k\leftarrow V_k\cup u_{\mu(k)}$}
        \ENDIF
    \ENDFOR
    \STATE{$V\leftarrow V\cup V_k$}
\ENDFOR
\RETURN{$\sigma=(\mathcal{S}_1,\mathcal{S}_2,\ldots,\mathcal{S}_{\ceil{\log n}},\mathcal{E})$} 
\end{algorithmic}
\end{algorithm}

\noindent\underline{\textbf{{\em \texttt{Star}} algorithm:}} The description of the algorithm can be found in Section \ref{sec:approach}.

\begin{lemma}
For any star graph $G$ on $n$ vertices, the {\em \texttt{star}} algorithm computes in polynomial time a growth schedule $\sigma$ for $G$ of $\ceil{\log{n}}$ slots and $O(n)$ excess edges.
\end{lemma}

\begin{proof}
    Let $u_0$ be the initiator and $n$ the size of the star graph. By construction, we can see that in every slot a star graph is maintained. In order to terminate with a star graph of size $n$, we require $\ceil{\log{n}}$ slots since for every vertex a new vertex is generated in each slot and therefore after $\ceil{\log{n}}$ slots, the graph will have size $n$. 
    
    For every vertex generation, there are at most two edge activations. Since there are $n-1$ vertices generated in total, there are $2(n-1)$ total edge activations. Therefore, the excess edges are at most $2(n-1)-(n-1)=O(n)$. Finally, note that if an excess edge is activated in slot $t$, then it is deleted in slot $t+1$ which results in maximum lifetime of $1$.
\end{proof}

\section{Growth Schedules of Zero Excess Edges}\label{sec:zero_excess}

In this section, we study which target graphs $G$ can be grown using $\ell=0$ excess edges for $d=2$. We begin by providing an algorithm that decides whether a graph $G$ can be grown by any schedule $\sigma$. We build on to that, by providing an algorithm that computes a schedule of $k=\log n$ slots for a target graph $G$, if one exists. We finish with our main technical result showing that computing the smallest schedule for a graph $G$ is NP-complete and any approximation of the shortest schedule cannot be within a factor of $ n^{\frac{1}{3}-\varepsilon}$  of the optimal solution, for any $\varepsilon>0$, unless $P=NP$. First, we check whether a graph $G$ has a growth schedule of $\ell=0$ excess edges. Observe that a graph $G$ has a growth schedule if and only if it has a schedule of $k=n-1$ slots.

\begin{definition}\label{def:candidates}
    Let $G = (V,E)$ be any graph. A vertex $v\in V$ can be the last generated vertex in a growth schedule $\sigma$ of $\ell=0$ for $G$ if there exists a vertex $w\in V\setminus\{v\}$ such that $N[v]\subseteq N[w]$. In this case, $v$ is called a {\em candidate} vertex and $w$ is called the \emph{candidate parent} of $v$. Furthermore, the set of candidate vertices in $G$ is denoted by $S_G = \{v\in V : N[v] \subseteq N[w] \text{ for some } w\in V\setminus \{v\}\}$\ see Figure~$\ref{fig:candidates}$.
\end{definition}

\begin{definition}
    A candidate elimination ordering of a graph G is an ordering $v_1,v_2, \ldots,v_n$ of $V(G)$ such that $v_i$ is a candidate vertex in the subgraph induced by $v_i, \ldots,v_n$, for $1\leq i\leq n$.
\end{definition}

\begin{figure}[h]
	\centering
	\begin{minipage}{0.23\textwidth}
		\centering
		\includegraphics[width=1\linewidth]{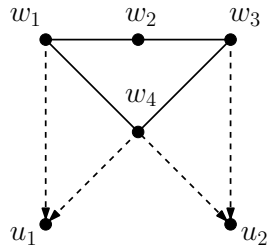}
	\end{minipage}\qquad

	\caption{Consider the above graph $G_t$ to be the graph grown after slot $t$. Vertices $u_1$ and $u_2$ are candidate vertices. The arrows represent all possible vertex generations in the previous slot $t$. Vertices $w_1$ and $w_4$ are candidate parents of $u_1$, while $w_3$ and $w_4$ are candidate parents of $u_2$.}
	\label{fig:candidates}
\end{figure}

\begin{lemma}\label{lem:cop-win}
    A graph $G$ has a growth schedule of $n-1$ slots and $\ell=0$ excess edges if and only if $G$ has a candidate elimination ordering.
\end{lemma}

\begin{proof}
    By definition of the model, whenever a vertex $u$ is generated for a vertex $w$ in a slot $t$, only edges between $u$ and vertices in $N[w]$ can be activated, which means that $N[u]\subseteq N[w]$. Since $\ell=0$, this property stays true in $G_{t+1}$. Therefore, any vertex $u$ generated in slot $t$, is a candidate vertex in graph $G_{t+1}$.
\end{proof}

The following algorithm can decide whether a graph has a candidate elimination ordering, and therefore, whether it can be grown with a schedule of $n-1$ slots and $\ell=0$ excess edges. The algorithm computes the slots of the schedule in reverse order. 

\medskip

\noindent\underline{\textbf{{\em \texttt{Candidate elimination ordering}} algorithm}:} Given the graph $G=(V,E)$, the algorithm finds all candidate vertices and deletes an arbitrary candidate vertex and its incident edges. The deleted vertex is added in the last empty slot of the schedule $\sigma$. The algorithm repeats the above process until there is only a single vertex left. If that is the case, the algorithm produces a growth schedule. If the algorithm cannot find any candidate vertex for removal, it decides that the graph cannot be grown. 

\begin{algorithm}[t]
\caption{Candidate elimination order}
\vspace{0.2cm}
\label{cop-win-alg}
\begin{algorithmic}[1]
\REQUIRE{A graph $G=(V,E)$ on $n$ vertices.}
\ENSURE{A growth schedule for $G$, if it exists}

\FOR{$k=n-1$ downto $1$}
    \STATE{$\mathcal{S}_k=\emptyset$}
    \FOR{every vertex $v\in V$}
        \IF[$v$ is a new candidate vertex]{($N[v]\subseteq N[u]$, for some vertex $u\in V\setminus \{v\}) \land (\mathcal{S}_k = \emptyset)$}
            \STATE{$\mathcal{S}_k \leftarrow \{(u,v,\{vw: w\in N(v)\})\}$}
            \STATE{$V \leftarrow V\setminus\{v\}$}
            
        \ENDIF
    \ENDFOR
    \IF{$\mathcal{S}_k = \emptyset$}
        \RETURN{``NO''}
    \ENDIF
\ENDFOR
\RETURN{$\sigma=(\mathcal{S}_1,\mathcal{S}_2,\ldots,\mathcal{S}_{n-1}, \emptyset)$}
\end{algorithmic}
\end{algorithm}

\begin{lemma}\label{lem:removing-candidates}
    Let $v\in S_G$. Then $G$ has a candidate elimination ordering if and only if $G-v$ has a candidate elimination ordering.
\end{lemma}

\begin{proof}
    Let $c$ be a candidate elimination ordering of $G-v$. Then, generating vertex $v$ at the end of $c$ trivially results in a candidate elimination ordering of $G$.
    
    Conversely, let $c$ be a candidate elimination ordering of $G$. If $v$ is the last vertex in $c$, then $c\setminus \{v\}$ is trivially a candidate elimination ordering of $G-v$. Suppose that the last vertex of $c$ is a vertex $u\neq v$. As $v\in S_G$ by assumption, there exists a vertex $w\neq v$ such that $N[v]\subseteq N[w]$. If $v$ does not give birth to any vertex in $c$ then we can move $v$ to the end of $c$, i.e.,~right after vertex $u$. Let $c'$ be the resulting candidate elimination ordering of $G$; then $c'\setminus v$ is a candidate elimination ordering of $G-v$, as the parent-child relations of $G-v$ are the same in both $c'\setminus v$ and $c$.

    Finally suppose that $v$ gives birth to at least one vertex, and let $Z$ be the set of vertices which are born by $v$ or by some descendant of $v$. 
    If $w$ appears before $v$ in $c$, then for any vertex in $Z$ we assign its parent to be $w$ (instead of $v$). Note that this is always possible as $N[v]\subseteq N[w]$. Now suppose that $w$ appears after $v$ in $c$, and let $Z_0 = \{z\in Z : v<_c z <_c w\}$ be the vertices of $Z$ which lie between $v$ and $w$ in $c$. Then we move all vertices of $Z_0$ immediately after $w$ (without changing their relative order). Finally, similarly to the above, for any vertex in $Z$ we assign its parent to be $w$ (instead of $v$). In either case (i.e.,~when $w$ is before or after $v$ in $c$), after making these changes we obtain a candidate elimination ordering $c''$ of $G$, in which $v$ does not give birth to any other vertex. Thus we can obtain from $c''$ a new candidate elimination ordering $c'''$ of $G$ where $v$ is moved to the end of the ordering. Then $c'''\setminus v$ is a candidate elimination ordering of $G-v$, as the parent-child relations of $G-v$ are the same in both $c'''\setminus v$ and $c''$. 
\end{proof}

\begin{theorem}
    The {\em \texttt{candidate elimination ordering}} algorithm is a polynomial-time algorithm that, for any graph $G$, decides whether $G$ has a growth schedule of $n-1$ slots and $\ell=0$ excess edges, and it outputs such a schedule if one exists.
\end{theorem}

\begin{proof}
    First note that we can find the candidate vertices in polynomial time, and thus, the algorithm terminates in polynomial time. This is because the algorithm removes one candidate vertex $u$ in each loop, which based on Lemma \ref{lem:removing-candidates}. By reversing the ordering of the removed vertices the algorithm can produce a growth schedule for $G$ if one exists.
\end{proof}

The notion of candidate elimination orderings turns out to coincide with the notion of cop-win orderings, discovered in the past in graph theory for a class of graphs, called cop-win graphs \cite{LinSS12,Bandelt91,poston71}. In particular, it is not hard to show that \emph{a graph has a candidate elimination ordering if and only if it is a cop-win graph}. This implies that our {\em \texttt{candidate elimination ordering}} algorithm is probably equivalent to some folklore algorithms in the literature of cop-win graphs. 

\begin{lemma}
    There is a modified version of the candidate elimination ordering algorithm that computes in polynomial time a growth schedule for any graph $G$ of $n-1$ slots and $\ell$ excess edges, where $\ell$ is a constant, if and only if such a schedule exists.
\end{lemma}

\begin{proof}
    The candidate elimination ordering algorithm can be slightly modified to check whether a graph $G=(V,E)$ has a growth schedule of $n-1$ slots and $\ell$ excess edges. The modification is quite simple. For $\ell=1$, we create multiple graphs $G'_x$ for $x=1, 2, \ldots, \frac{n(n-1)}{2} - |E|$ where each graph $G'_x$ is a copy of $G$ with the addition of one edge $e\notin E$, and we do this for all possible edge additions. In essence, we create $G'_x=(V'_x,E'_x)$, where $V'_x=V$ and $E'_x=E\cup uv$ such that $uv\not\in E$ and $(E'_j \neq E'_i)$, for all $i \neq j$. Since the complement of $G$ has $e\leq \frac{n(n-1)}{2}$ edges, we will create up to $\frac{n(n-1)}{2}$ graphs $G'_x$. We then run the candidate elimination ordering algorithm on all $G'_x$. If the algorithm returns ``no'' for all of them, then there exists no growth schedule for $G$ of $n-1$ slots and $1$ excess edge. Otherwise, the algorithm outputs a schedule of $n-1$ slots and $1$ excess edge for graph $G$. This process can be  modified for any $\ell$, but then the number of graphs $G'_x$ tested is at most $\frac{n^{\ell}(n-1)}{2}$. Therefore if $\ell$ is a constant, all graphs $G'_x$ can be checked in polynomial time.
\end{proof}

Our next goal is to decide whether a graph $G=(V,E)$ on $n$ vertices has a  growth schedule $\sigma$ of $\log n$ slots and $\ell=0$ excess edges. The {\em \texttt{fast growth}} algorithm computes the slots of the growth schedule in reverse order. 

\medskip
\noindent\underline{\textbf{{\em \texttt{Fast growth}} algorithm}:} The algorithm finds set $S_G$ of candidate vertices in $G$. It then tries to find a subset $L\subseteq S_G$ of candidates that satisfies all of the following properties:

\begin{enumerate}
   \item $|L|=n/2$. 
   \item $L$ is an independent set. 
   \item There is a perfect matching between the candidate vertices in $L$ and their candidate parents in $G$.
\end{enumerate}

Any set $L$ that satisfies the above constraints is called {\em valid}. The algorithm finds such a set by creating a 2-SAT formula $\phi$ whose solution is a valid set $L$. If the algorithm finds such a set $L$, it adds the vertices in $L$ to the last slot of the schedule. It then removes the vertices in $L$ from graph $G$ along with their incident edges. The above process is then repeated to find the next slots. If at any point, graph $G$ has a single vertex, the algorithm terminates and outputs the schedule. If at any point, the algorithm cannot find a valid set $L$, it outputs ``no''.

\begin{algorithm}[t!]
\caption{Fast growth algorithm}
\vspace{0.2cm}
\label{fast-cop-win-alg}
\begin{algorithmic}[1]
\REQUIRE{A graph $G=(V,E)$ on $n=2^\delta$ vertices.}
\ENSURE{A growth schedule of $k=\log n$ slots and $\ell=0$ excess edges for $G$.}
\FOR{$k=\log n$ downto $1$}
    \STATE{$\mathcal{S}_k=\emptyset$; \ \ $\phi=\emptyset$}
    \STATE{Find a perfect matching $M=\{u_iv_i:1\leq i \leq n/2\}$ of $G$.}
    \IF{No perfect matching exists}
        \RETURN{"NO"}
    \ENDIF    
    \FOR {every edge $u_iv_i \in M$}
        \STATE{Create variable $x_i$}
    \ENDFOR    
    \FOR{every edge $u_iv_i \in M$}   
        \IF[$u_i$ is a candidate vertex and $v_i$ is not.]{$(N[u_i]\subseteq N[w]$, for some vertex $w\in V\setminus \{u_i\}) \ \land \ (N[v_i]\not\subseteq N[x]$, for any vertex $x\in V\setminus \{v_i\})$} 
            \STATE{$\phi \leftarrow \phi \land {(\neg x_i)}$}
        \ELSIF[$u_i$ is a not candidate and $v_i$ is a candidate]{($N[u_i]\not\subseteq N[w]$ for any vertex $w\in V\setminus \{u_i\}) \  \land \ (N[v_i]\subseteq N[x]$, for some vertex $x\in V\setminus \{v_i\})$} 
            \STATE{$\phi \leftarrow \phi \land (x_i)$}
        \ELSIF[$u_i$ is not a candidate and $v_i$ is not a candidate]{($N[u_i]\not\subseteq N[w]$, for some vertex $w\in V\setminus \{u_i\}) \ \land \ (N[v_i]\not\subseteq N[x]$, for some vertex $x\in V\setminus \{v_i\})$} 
            \RETURN{"NO"}
        \ENDIF        
    \ENDFOR   
    \FOR{every edge $u_iu_j\in E\setminus M$}
        \STATE{$\phi \leftarrow \phi \land (x_i\lor x_j)$}
    \ENDFOR     
    \FOR{every edge $v_iv_j\in E\setminus M$}
        \STATE{$\phi \leftarrow \phi \land (\neg{x_i}\lor {\neg x_j})$}
    \ENDFOR 
    \FOR{every edge $u_iv_j\in E\setminus M$}
        \STATE{$\phi \leftarrow \phi \land (x_i\lor {\neg x_j})$}
    \ENDFOR
    
    \IF{$\phi$ is satisfiable}
        \STATE{Let $\tau$ be a satisfying truth assignment for $\phi$}
        \FOR{$i=1,2,\ldots,n/2$}
            \IF{$x_i=true$ in $\tau$}
                \STATE{$\mathcal{S}_k\leftarrow \mathcal{S}_k\cup (u_i,v_i,\{v_iw: w\in N(v_i)\})$}
                \STATE{$V \leftarrow V\setminus\{v_i\}$}
                \STATE{$E \leftarrow E\setminus\{v_iw: w\in N(v_i)\}$}
            \ELSE[$x_i=false$ in $\tau$]
                \STATE{$\mathcal{S}_k\leftarrow \mathcal{S}_k\cup (v_i,u_i,\{u_iw: w\in N(u_i)\})$}
                \STATE{$V \leftarrow V\setminus\{u_i\}$}
                \STATE{$E \leftarrow E\setminus\{u_iw: w\in N(u_i)\}$}
            \ENDIF    
        \ENDFOR
    \ELSE[$\phi$ is not satisfiable]
        \RETURN{"NO"}
    \ENDIF
\ENDFOR
\RETURN{$\sigma=(\mathcal{S}_1,\mathcal{S}_2,\ldots,\mathcal{S}_{k},\emptyset)$}

\end{algorithmic}
\end{algorithm}

\begin{lemma}\label{lem:perfect_matching1}
    Consider any graph $G=(V,E)$. If $G$ has a growth schedule of $\log n$ slots and $\ell=0$ excess edges then there exists a perfect matching $M$ that contains a valid candidate vertex set $L$, where $L$ has exactly one vertex for each edge of the perfect matching $M$.
\end{lemma}

\begin{proof}
    Let us assume that graph $G$ has a growth schedule. Then in the last slot, there are $n/2$ vertices, called parents, for which $n/2$ other vertices, called children, are generated. Therefore, such a perfect matching $M$ always exists where set $L$ contains the children.
\end{proof}

\begin{lemma}\label{lem:2SAT}
    The 2-SAT formula $\phi$, generated by the {\em \texttt{fast growth algorithm}}, has a solution if and only if there is an independent set $|V_2|=n/2$, where $V_2$ is a valid set of candidate vertices in graph $G=(V,E)$.
\end{lemma}

\begin{proof}
    Let us assume that graph $G$ has a growth schedule. Based on Lemma $\ref{lem:perfect_matching1}$, there are $n/2$ parents and $n/2$ children in $G$. Therefore, there has to be a set $V_2$, where $|V_2|=n/2$ and $V_2$ is an independent set such that there is another set $V_1$, where $|V_1|=n/2$ and $V_1\cap V_2=\emptyset$. Any perfect matching $M \in G$ includes edges $u_iv_i\in M$, where $u_i\in V_1$ and $v_i\in V_2$ because $V_2$ is an independent set. 
    
    The solution to the $2$-SAT formula $\phi$ we are going to create is a valid set $V_2$ as stated above. Consider an arbitrary edge $u_iv_i$ from the perfect matching $M$. The algorithm creates a variable $x_i$ for each $u_iv_i$. The truthful assignment of $x_i$ means that we pick $v_i$ for $V_2$ and the negative assignment means that we pick $u_i$ for $V_2$. Since $|V_2|=n/2$, then for every edge $u_iv_i\in M$, at least one of $u_i,v_i$ is a candidate vertex, because otherwise some other edge $u_jv_j\in M$ would need to have $2$ candidates vertices at its endpoints and include them both in $V_2$, which is impossible. Thus, graph $G$ would have no growth schedule.
    
    If $v_i$ is a candidate vertex and $u_i$ is not, then $v_i\in V_2$, and we add clause $(x_i)$ to $\phi$. 
    If $u_i$ is a candidate vertex and $v_i$ is not, then $u_i\in V_2$, in which case we add clause $(\neg x_i)$ $\phi$.
    If both $u_i$ and $v_i$ are candidate vertices, either one could be in $V_2$ as long as $V_2$ is an independent set.
    
    We now want to make sure that every vertex in $V_2$ is independent. Therefore, for every edge $u_iu_j\in E$, we add clause $(x_i\lor x_j)$ to $\phi$. This means that in order to satisfy that clause, $u_i$ and $u_j$ cannot be both picked for $V_2$. Similarly, for every edge $v_iv_j\in E$, we add clause $(\neg x_i)\lor (\neg x_j)$ to $\phi$ and for every edge $u_iv_j\in E$, we add clause $(x_i)\lor (\neg x_j)$ to $\phi$.
    
    The solution to formula $\phi$ is a valid set $V_2$ and we can find it in polynomial time. If the formula has no solution, then no valid independent set $V_2$ exists for graph $G$.
\end{proof}

\begin{lemma}\label{lem:perfect_matching}
    Consider any graph $G=(V,E)$. If $G$ has a growth schedule of $\log n$ slots and $\ell=0$ excess edges, then any arbitrary perfect matching contains a valid candidate set $|L|=n/2$, where $L$ has exactly one vertex for each edge of the perfect matching.
\end{lemma}

\begin{proof}
    By Lemma $\ref{lem:2SAT}$, any perfect matching $M$ contains edges $uv$, such that there exists a valid candidate set $V_2$ that contains one vertex exactly for each edge $uv\in M$. Thus, if graph $G$ has a growth schedule, the solution to the $2$-SAT formula corresponds to a valid candidate set~$V_2$.
\end{proof}

\begin{theorem}
    For any graph $G$ on $2^\delta$ vertices, the {\em \texttt{fast growth}} algorithm computes in polynomial time a growth schedule $\sigma$ for $G$ of $\log n$ slots and $\ell=0$ excess edges, if one exists.
\end{theorem}

\begin{proof}
    Suppose that $G=(V,E)$ has a growth schedule $\sigma$ of $\log n$ slots and $\ell=0$ excess edges.
    By Lemmas $\ref{lem:2SAT}$ and $\ref{lem:perfect_matching}$ we know that our {\em \texttt{fast growth}} algorithm finds a set $L$ for the last slot of a schedule $\sigma''$ but this might be a different set from the last slot contained in $\sigma$. Therefore, for our proof to be complete, we need to show that if $G$ has a growth schedule $\sigma$ of $\log n$ slots and $\ell=0$ excess edges, for any $L$ it holds that $(G-L)$ has a growth schedule $\sigma'$ of $\log n - 1$ slots and $\ell=0$ excess edges.
    
    Assume that $\sigma$ has in the last slot $\mathcal{S}_{k}$ a set of vertices $V_1$ generating another set of vertices $V_2$, such that $|V_1|=|V_2|=n/2$, $V_1\cap V_2=\emptyset$ and $V_2$ is an independent set. Suppose that our algorithm finds $V'_2$ such that $V'_2\neq V_2$. 
    
    Assume that $V'_2\cap V_2=V_s$ and $|V_s|=n/2-1$. This means that $V'_2=V_s\cup u'$ and $V_2=V_s\cup u$ and $u'$ has no edge with any vertex in $V_s$. Since $u'\not\in V_2$ and $u'$ has no edge with any vertex in $V_s$, then $u'\in V_1$. However, $u'$ cannot be the candidate parent of anyone in $V_2$ apart from $u$. Similarly, $u$ is the only candidate parent of $u'$. Therefore $N[u]\subseteq N[u']\subseteq N[u] \implies N[u] = N[u']$. This means that we can swap the two vertices in any growth schedule and still maintain a correct growth schedule for $G$. Therefore, for $L=V'_2$, the graph $(G-L)$ has a growth schedule $\sigma'$ of $\log n - 1$ slots and $\ell=0$ excess edges. 
    
    Assume now that $V'_2\cap V_2=V_s$, where $|V_s| = x \geq 0$. Then, $V'_2=V_s\cup u'_1\cup u'_2,\cup \ldots\cup u'_y$ and $V_2=V_s\cup u_1\cup u_2,\cup \ldots\cup u_y$, where $y=n/2-x$. As argued above, vertices $u'_1,u'_2,\ldots,u'_y$ can be candidate parents only to vertices $u_1,u_2,\ldots,u_y$, and vice versa. Thus, there is a pairing $u_j, u_j'$ such that $N[u_j]\subseteq N[u'_j]\subseteq N[u_j] \implies N[u'_j] = N[u_j]$, for every $j=1,2,\ldots,y$. Thus, these vertices can be swapped in the growth schedule and still maintain a correct growth schedule for $G$. Therefore for any arbitrary $L=V'_2$, the graph $(G-L)$ has a growth schedule $\sigma'$ of $\log n - 1$ slots and $\ell=0$ excess edges.
\end{proof}

We will now show that the problem of computing the minimum number of slots required for a graph $G$ to be grown is NP-complete, and that it cannot be approximated within a $n^{\frac{1}{3}-\varepsilon}$ factor for any $\varepsilon>0$, unless P~=~NP.

\begin{definition}
Given any graph $G$ and a natural number $\kappa$, find a growth schedule of $\kappa$ slots and $\ell=0$ excess edges. We call this problem \emph{zero-excess growth schedule}.
\end{definition}

\begin{theorem}\label{thm:NP-HARD}
    The decision version of the zero-excess graph growth problem is NP-complete.
\end{theorem}

\begin{proof}
    First, observe that the decision version of the problem belongs to the class NP. Indeed, the required polynomial certificate is a given growth schedule $\sigma$, together with an isomorphism between the graph grown by $\sigma$ and the target graph $G$.

    To show NP-hardness, we provide a reduction from the coloring problem. Given an arbitrary graph $G=(V,E)$ on $n$ vertices, we grow graph $G'=(V',E')$ as follows: 
    Let $G_1=(V_1,E_1)$ be an isomorphic copy of $G$, and let $G_2$ be a clique of $n$ vertices. $G'$ consists of the union of $G_1=(V_1,E_1)$ and $G_2=(V_2,E_2)$, where we also add all possible edges between them. Note that every vertex of $G_2$ is a universal vertex in $G'$ (i.e.,~a vertex which is connected with every other vertex in the graph). 
    Let $\chi(G)$ be the chromatic number of graph $G$, and let $\kappa(G')$ be the minimum number of slots required for a growth schedule for $G'$. We will show that $\kappa(G')= \chi(G)+n$.
    
    Let $\sigma$ be an optimal growth schedule for $G'$, which uses $\kappa(G')$ slots. As every vertex $v\in V_2$ is a universal vertex in $G'$, $v$ cannot coexist with any other vertex of $G'$ in the same slot of $\sigma$. Furthermore, the vertices of $V_1$ require at least $\chi(G)$ different slots in $\sigma$, since $\chi(G)$ is the smallest possible partition of $V_1$ into independent sets. Thus $\kappa(G')\geq \chi(G)+n$.
    
    We now provide the following growth schedule $\sigma^*$ for $G'$, which consists of exactly $\chi(G)+n$ slots. Each of the first $n$ slots of $\sigma^*$ contains exactly one vertex of $V_2$; note that each of these vertices (apart from the first one) can be generated and connected with an earlier vertex of $V_2$. In each of the following $\chi(G)$ slots, we add one of the $\chi(G)=\chi(G_1)$ color classes of an optimal coloring of $G_1$. Consider an arbitrary color class of $G_1$ and suppose that it contains $p$ vertices; these $p$ vertices can be born by exactly $p$ of the universal vertices of $V_2$ (which have previously appeared in $\sigma^*$). This completes the growth schedule $\sigma^*$. Since $\sigma^*$ has $\chi(G)+n$ slots, it follows that $\kappa(G')\leq \chi(G)+n$.
\end{proof}

\begin{theorem}
    Let $\varepsilon>0$. If there exists a polynomial-time algorithm, which, for every graph $G$, computes a $n^{\frac{1}{3} - \varepsilon}$-approximate growth schedule (i.e.,~a growth schedule with at most $n^{\frac{1}{3} - \varepsilon} \kappa(G)$ slots), then \textup{P}~$=$~\textup{NP}.
\end{theorem}

\begin{proof}
    The reduction is from the minimum coloring problem. Given an arbitrary graph $G=(V,E)$ with $n$ vertices, we construct in polynomial time a graph $G'=(V',E')$ with $N=4n^3$ vertices, as follows: We create $2n^2$ isomorphic copies of $G$, which are denoted by $G^A_1,G^A_2,\ldots,G^A_{n^2}$ and $G^B_1,G^B_2,\ldots,G^B_{n^2}$, and we also add $n^2$ clique graphs, each of size $2n$, denoted by $C_1,C_2,\ldots,C_{n^2}$. We define $V'=V(G^A_1)\cup \ldots\cup  V(G^A_{n^2}) \cup V(G^B_1)\cup \ldots\cup  V(G^B_{n^2})\cup V(C_1) \cup \ldots\cup  V(C_{n^2})$. Initially we add to the set $E'$ the edges of all graphs $G^A_1,\ldots,G^A_{n^2}$, $G^B_1,\ldots,G^B_{n^2}$, and $C_1,\ldots,C_{n^2}$. 
    For every $i=1,2,\ldots,n^2-1$ we add to $E'$ all edges between $V(G^A_i) \cup V(G^B_i)$ and $V(G^A_{i+1}) \cup V(G^B_{i+1})$. 
    For every $i=1,\ldots,n^2$, we add to $E'$ all edges between $V(C_{i})$ and $V(G^A_i) \cup V(G^B_i)$. Furthermore, for every $i=2,\ldots,n^2$, we add to $E'$ all edges between $V(C_{i})$ and $V(G^A_{i-1}) \cup V(G^B_{i-1})$.  For every $i=1,\ldots,n^2-1$, we add to $E'$ all edges between $V(C_{i})$ and $V(C_{i+1})$.
    For every $i=1,2,\ldots,n^2$ and for every $u\in V(G^B_i)$, we add to $E'$ the edge $uu'$, where $u'\in V(G^A_i)$ is the image of $u$ in the isomorphism mapping between $G^A_i$ and $G^B_i$. 
    To complete the construction, we pick an arbitrary vertex $a_i$ from each $C_i$. We add edges among the vertices $a_1,\ldots,a_{n^2}$ such that the resulting induced graph $G'[a_1,\ldots,a_{n^2}]$ is a graph on $n^2$ vertices which can be grown by a {\em \texttt{path}} schedule within $\lceil \log n^2 \rceil$ slots and with zero excess edges (see Lemma~\ref{line-construction-lem}\footnote{From Lemma~\ref{line-construction-lem} it follows that the path on $n^2$ vertices can be grown in $\lceil \log n^2 \rceil$ slots using $O(n^2)$ excess edges. If we put all these $O(n^2)$ excess edges back to the path of $n^2$ vertices, we obtain a new graph on $n^2$ vertices with $O(n^2)$ edges. This graph is the induced subgraph $G'[a_1,\ldots,a_{n^2}]$ of $G'$ on the vertices $a_1,\ldots,a_{n^2}$.}). 
    This completes the construction of~$G'$. Clearly, $G'$ can be grown in time polynomial in $n$.

    Now we will prove that there exists a growth schedule $\sigma'$ of $G'$ of length at most 
    $n^2 \chi(G)+ 4n-2 + \lceil 2\log n \rceil$. The schedule will be described inversely, that is, we will describe the vertices generated in each slot starting from the last slot of $\sigma'$ and finishing with the first slot. 
    First note that every $u\in V(G^A_{n^2})\cup V(G^B_{n^2})$ is a candidate vertex in $G'$ Indeed, for every $w\in V(C_{n^2})$, we have that $N[u]\subseteq V(G^A_{n^2})\cup V(G^B_{n^2})\cup V(G^A_{n^2-1})\cup V(G^A_{n^2-1})\cup V(C_{n^2}) \subseteq N[w]$. 
    To provide the desired growth schedule $\sigma'$, we assume that a minimum coloring of the input graph $G$ (with $\chi(G)$ colors) is known. In the last $\chi(G)$ slots, $\sigma'$ generates all vertices in $V(G^A_{n^2})\cup V(G^B_{n^2})$, as follows. At each of these slots, one of the $\chi(G)$ color classes of the minimum coloring $c_{OPT}$ of $G_{n^2}^A$ is generated on sufficiently many vertices among the first $n$ vertices of the clique $C_{n^2}$. Simultaneously, a different color class of the minimum coloring $c_{OPT}$ of $G_{n^2}^B$ is generated on sufficiently many vertices among the last $n$ vertices of the clique $C_{n^2}$.
    
    Similarly, for every $i=1,\ldots,n^2-1$, once the vertices of $V(G^A_{i+1})\cup \ldots \cup V(G^A_{n^2}) \cup V(G^B_{i+1})\cup \ldots \cup V(G^B_{n^2})$ have been added to the last $(n^2-i)\chi (G)$ slots of $\sigma'$, the vertices of $V(G^A_{i})\cup V(G^B_{i})$ are generated in $\sigma'$ in $\chi(G)$ more slots. This is possible because every vertex $u\in V(G^A_{i})\cup V(G^B_{i})$ is a candidate vertex after the vertices of $V(G^A_{i+1})\cup \ldots \cup V(G^A_{n^2}) \cup V(G^B_{i+1})\cup \ldots \cup V(G^B_{n^2})$ have been added to slots. Indeed, for every $w\in V(C_{i})$, we have that $N[u]\subseteq V(G^A_{i})\cup V(G^B_{i})\cup V(G^A_{i-1})\cup V(G^A_{i-1})\cup V(C_{i}) \subseteq N[w]$. 
    That is, in total, all vertices of $V(G^A_{1})\cup \ldots \cup V(G^A_{n^2}) \cup V(G^B_{1})\cup \ldots \cup V(G^B_{n^2})$ are generated in the last $n^2 \chi(G)$ slots. 
    
    The remaining vertices of $V(C_1)\cup \ldots \cup V(C_{n^2})$ are generated in $\sigma'$ in $4n-2 +\lceil\log n^2 \rceil$ additional slots. First, for every odd index $i$ and for $2n-1$ consecutive slots, for vertex $a_i$ of $V(C_i)$ exactly one other vertex of $V(C_i)$ is generated. This is possible because for every vertex $u\in V(C_i)\setminus a_i$, $N[u]\subseteq V(C_i) \cup V(C_{i-1}) \cup V(C_{i+1})\subseteq N[a_i]$. Then, for every even index $i$ and for $2n-1$ further consecutive slots, for vertex $a_i$ of $V(C_i)$ exactly one other vertex of $V(C_i)$ is generated. That is, after $4n-2$ slots only the induced subgraph of $G'$ on the vertices $a_1,\ldots,a_{n^2}$ remains. The final $\lceil \log n^2 \rceil$ slots of $\sigma'$ are the ones obtained by Lemma~\ref{line-construction-lem}. 
    To sum up, $G'$ is grown by the growth schedule $\sigma'$ in $k=n^2 \chi(G) + 4n-2 + \lceil \log n^2 \rceil$ slots, and thus 
    \begin{align}\label{equ:1}
    \kappa(G')\leq n^2 \chi(G) + 4n-2 + \lceil 2\log n \rceil
    \end{align}
    
    Suppose that there exists a polynomial-time algorithm $A$ which computes an $N^{\frac{1}{3}-\varepsilon}$-approximate growth schedule $\sigma''$ for graph $G'$ (which has $N$ vertices), i.e.,~a growth schedule of $k\leq N^{\frac{1}{3} - \varepsilon} \kappa(G')$ slots. 
    Note that, for every slot of $\sigma''$, all different vertices of $V(G^A_i)$ (resp.~$V(G^B_i)$) which are generated in this slot are independent. 
    For every $i=1,\ldots,n^2$, denote by $\chi^A_i$ (resp.~$\chi^B_i$) the number of different slots of $\sigma''$ in which at least one vertex of $V(G^A_i)$ (resp.~$V(G^B_i)$) appears. Let $\chi^* = \min\{\chi^A_i,\chi^B_i : 1\leq i \leq n^2\}$. Then, there exists a coloring of $G$ with at most $\chi^*$ colors (i.e.,~a partition of $G$ into at most $\chi^*$ independent sets).
    
    Now we show that $k \geq  \frac{1}{2}n^2 \chi^*$. 
    Let $i\in \{2,\ldots,n^2-1\}$ and let $u\in V(G^A_{i})\cup V(G^B_{i})$.
    Assume that $u$ is generated at slot $t$ in $\sigma''$. Then, either all vertices of $V(G^A_{i-1})\cup V(G^B_{i-1})$ or all vertices of $V(G^A_{i+1})\cup V(G^B_{i+1})$ are generated at a later slot $t'\geq t+1$ in $\sigma''$. 
    Indeed, it can be easily checked that, if otherwise both a vertex $x\in V(G^A_{i-1})\cup V(G^B_{i-1})$ and a vertex $y\in V(G^A_{i+1})\cup V(G^B_{i+1})$ are generated at a slot $t''\leq t$ in $\sigma''$, then $u$ cannot be a candidate vertex at slot $t$, which is a contradiction to our assumption. 
    That is, in order for a vertex $u\in V(G^A_{i})\cup V(G^B_{i})$ to be generated at some slot $t$ of $\sigma''$, we must have that $i$ is either the currently smallest or largest index for which some vertices of $V(G^A_{i})\cup V(G^B_{i})$ have been generated until slot $t$. 
    On the other hand, by definition of $\chi^*$, the growth schedule $\sigma''$ needs at least $\chi^*$ different slots to generate all vertices of the set $V(G^A_{i})\cup V(G^B_{i})$, for $1\leq i \leq n^2$. Therefore, since at every slot, $\sigma''$ can potentially generate vertices of at most two indices $i$ (the smallest and the largest respectively), it needs to use at least $\frac{1}{2}n^2 \chi^*$ slots to grow the whole graph $G'$. Therefore 
    \begin{align}\label{equ:2}
    k\geq \frac{1}{2}n^2 \chi^*
    \end{align}
    
    Recall that $N=4n^3$. It follows by Eq.\ref{equ:1} and Eq.\ref{equ:2} that 
    \begin{eqnarray*}
\frac{1}{2}n^{2}\chi ^{\ast } &\leq &k\leq N^{\frac{1}{3}-\varepsilon
}\kappa (G^{\prime }) \\
&\leq &N^{\frac{1}{3}-\varepsilon }(n^{2}\chi (G)+4n-2+\lceil 2\log n\rceil )
\\
&\leq &4n^{1-3\varepsilon }(n^{2}\chi (G)+6n)
\end{eqnarray*}
    and thus
    $\chi^* \leq 8n^{1-3\varepsilon} \chi(G) + 48n^{-3\varepsilon}$.
    Note that, for sufficiently large $n$, we have that $8n^{1-3\varepsilon} \chi(G) + 48n^{-3\varepsilon} \leq n^{1-\varepsilon} \chi(G)$. That is, given the $N^{\frac{1}{3}-\varepsilon}$-approximate growth schedule produced by the polynomial-time algorithm $A$, we can compute in polynomial time a coloring of $G$ with $\chi^*$ colors such that $\chi^* \leq n^{1-\varepsilon} \chi(G)$. This is a contradiction since for every $\varepsilon>0$, there is no polynomial-time $n^{1-\varepsilon}$-approximation for minimum coloring, unless P~=~NP~\cite{zuckerman2007linear}.
\end{proof}

\section{Growth Schedules of (Poly)logarithmic Slots}\label{sec:algorithms_basic_graph_classes}

In this section, we study graphs that have growth schedules of (poly)logarithmic slots, for $d=2$. As we have proven in the previous section, an integral factor in computing a growth schedule for any graph $G$, is computing a $k$-coloring for $G$. Since we consider polynomial-time algorithms, we have to restrict ourselves to graphs where the $k$-coloring problem can be solved in polynomial time and, additionally, we want small values of $k$ since we want to produce fast growth schedules. Therefore, we investigate tree, planar and $k$-degenerate graph families since there are polynomial-time algorithms that solve the $k$-coloring problem for graphs drawn from these families. We continue with lower bounds on the number of excess edges if we fix the number of slots to $\log n$, for path, star and specific bipartite graph families. 
\subsection{Trees}

We provide an algorithm that computes growth schedules for tree graphs. Let $G$ be the target tree graph. The algorithm applies a decomposition strategy on $G$, where vertices and edges are removed in phases, until a single vertex is left.
We can then grow the target graph $G$ by reversing its decomposition phases, using the {\em \texttt{path}} and {\em \texttt{star}} schedules as subroutines. 

\medskip

\noindent\underline{\textbf{{\em \texttt{Tree}} algorithm}:} Starting from a tree graph $G$, the algorithm keeps alternating between two phases, a \emph{path-cut} and a \emph{leaf-cut} phase. Let $G_{2i}$, $G_{2i+1}$, for $i\geq 0$, be the graphs obtained after the execution of the first $i$ pairs of phases and an additional path-cut phase, respectively.

\medskip

\noindent\textbf{Path-cut phase:} For each path subgraph $P=(u_1, u_2, \ldots, u_\nu)$, for $2<\nu\leq n$, of the current graph $G_{2i}$, where $u_2,u_3,...,u_{\nu-1}$ have degree $2$ and $u_1,u_\nu$ have degree $\neq 2$ in $G_{2i}$,
edge $u_1u_\nu$ between the endpoints of $P$ is activated and vertices $u_2,u_3,...u_{\nu-1}$ are removed along with their incident edges. An example of this is shown in Figure $\ref{fig:LineCut}$. 
If a single vertex is left, the algorithm terminates; otherwise, it proceeds to the leaf-cut phase.

\noindent\textbf{Leaf-cut phase:} Every leaf vertex of the current graph $G_{2i+1}$ is removed along with its incident edge. An example of this is shown in Figure $\ref{fig:LeafCut}$. 
If a single vertex is left, the algorithm terminates; otherwise, it proceeds to the path-cut phase.\\

Finally, the algorithm reverses the phases (by decreasing $i$) to output a growth schedule for the tree $G$ as follows. For each path-cut phase $2i$, all path subgraphs that were decomposed in phase $i$ are regrown by using the {\em \texttt{path}} schedule as a subprocess. These can be executed in parallel in $O(\log n)$ slots. The same holds true for leaf-cut phases $2i+1$, where each can be reversed to regrow the removed leaves by using {\em \texttt{star}} schedules in parallel in $O(\log n)$ slots. In the last slot, the schedule deletes every excess edge.

\begin{figure}[htb]
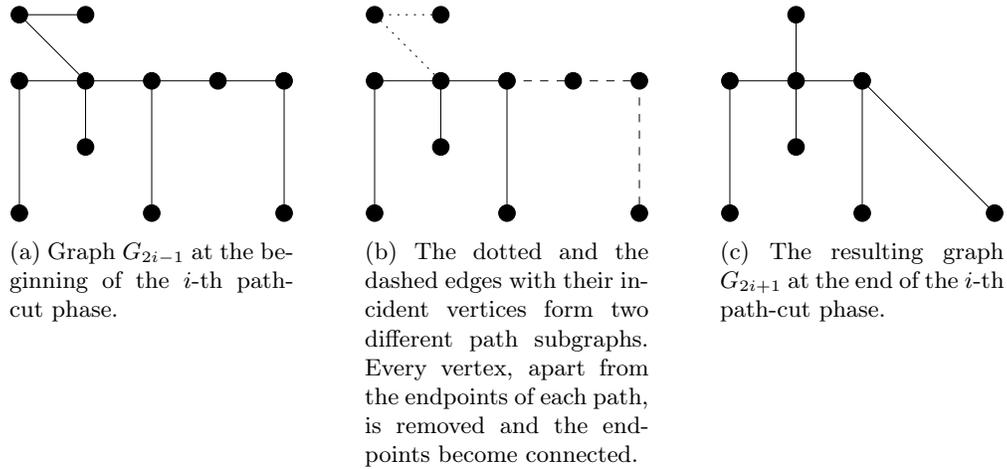

    \centering 
\begin{subfigure}[t]{0.25\textwidth}
  \includegraphics[width=\linewidth]{images/Efficient_Tree.pdf}
  \caption{Graph $G_{2i-1}$ at the beginning of the $i$-th path-cut phase.}
\end{subfigure}\hfil 
\begin{subfigure}[t]{0.25\textwidth}
  \includegraphics[width=\linewidth]{images/LIneCut.pdf}
  \caption{The dotted and the dashed edges with their incident vertices form two different path subgraphs. Every vertex, apart from the endpoints of each path, is removed and the endpoints become connected.}
\end{subfigure}\hfil 
\begin{subfigure}[t]{0.25\textwidth}
  \includegraphics[width=\linewidth]{images/LIneCut2.pdf}
  \caption{The resulting graph $G_{2i+1}$ at the end of the $i$-th path-cut phase.}
\end{subfigure}
\caption{An example of a path-cut phase.}\label{fig:LineCut}
\end{figure}

\begin{figure}[htb]
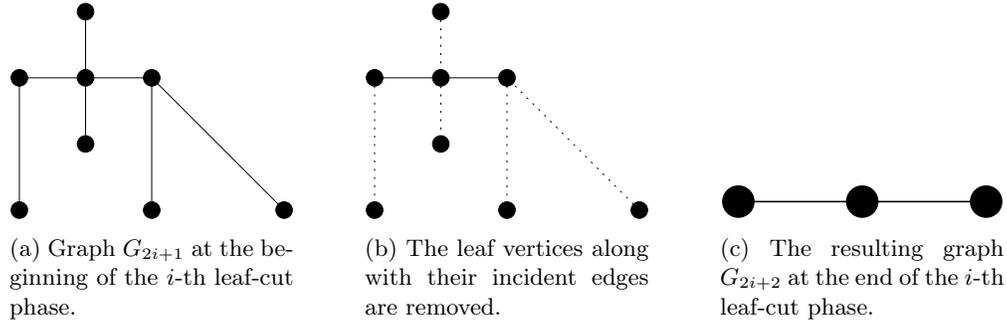

    \centering 
\begin{subfigure}[b]{0.25\textwidth}
  \includegraphics[width=\linewidth]{images/LIneCut2.pdf}
  \caption{Graph $G_{2i+1}$ at the beginning of the $i$-th leaf-cut phase.}
\end{subfigure}\hfil 
\begin{subfigure}[b]{0.25\textwidth}
  \includegraphics[width=\linewidth]{images/LeafCut.pdf}
  \caption{The leaf vertices along with their incident edges are removed.}
\end{subfigure}\hfil 
\begin{subfigure}[b]{0.25\textwidth}
  \includegraphics[width=\linewidth]{images/LeafCut2.pdf}
  \caption{The resulting graph $G_{2i+2}$ at the end of the $i$-th leaf-cut phase.}
\end{subfigure}
\caption{An example of a leaf-cut phase.}\label{fig:LeafCut}
\end{figure}

\begin{lemma}
    Given any tree graph $G$, the algorithm  deconstructs $G$ into a single vertex using $O(\log n)$ phases.
\end{lemma}

\begin{proof}
    Consider the graph $G_{2i}$ after the execution of the $i$-th path-cut phase. The path-cut phase removes every vertex that has exactly $2$ neighbors in the current graph, and in the next leaf cut phase, the graph consists of leaf vertices $u\in S_u$ with $deg=1$ and internal vertices $v\in S_v$ with $deg > 2$. Therefore, $|S_u|>|S_v|$ and since $|S_u|+|S_v|=|V_i|$, we can conclude that $|S_u|>|V_i|/2$ and any leaf-cut phase cuts the size of the current graph in half since it removes every vertex $u\in S_u$. This means that after at most $\log n$ path-cut phases and $\log n$ leaf cut phases the graph will have a single vertex.
\end{proof}

\begin{lemma} \label{lem:tree_slots}
    Every phase can be reversed using a growth schedule of $O(\log n)$ slots.
\end{lemma}

\begin{proof}
    First, let us consider the path-cut phase. At the beginning of this phase, every starting subgraph $G'$ is a path subgraph with vertices $u_1,u_2,...,u_x$, where $u_1,u_x$ are the endpoints of the path. At the end of the phase, every subgraph has two connected vertices $u_1,u_x$. The reversed process works as follows: for each path $u_1,u_2,\ldots u_x$ that we want to generate, we use vertex $u_1$ as the initiator and we execute the {\em \texttt{path}} algorithm from Section \ref{sec:edge_act_distance_2} in order to generate vertices $u_2,u_3,...,u_{x-1}$. We add the following modification to {\em \texttt{path}}: every time a vertex is generated, an edge between it and vertex $u_x$ is activated. After this process completes, edges not belonging to the original path subgraph $G'$ are deleted. This growth schedule requires $\log x \leq \log n$ slots. We can combine the growth schedules of each path into a single schedule of $\log x$ slots since every schedule has distinct initiators and they can run in parallel.
    
    Now let us consider the leaf-cut phase. In this phase, every vertex removed is a leaf vertex $u$ with one neighbor $v$. Note that $v$ might have multiple neighboring leaves. The reverse process works as follows: For each vertex $v$, we use a separate star growth schedule from Section \ref{sec:edge_act_distance_2} with $v$ as the initiator, in order to generate every vertex $u$ that was a neighbor to $v$. Each of this growth schedule requires at most $\log x \leq \log n$ slots, where $x$ is the number of leaves in the current graph. We can combine the growth schedules of each star into a single schedule of $\log k$ slots since every schedule has distinct initiators and they can run in parallel.
\end{proof}

\begin{theorem}
For any tree graph $G$ on $n$ vertices, the {\em \texttt{tree}} algorithm computes in polynomial time a growth schedule $\sigma$ for $G$ of $O(\log^2n)$ slots and $O(n)$ excess edges.
\end{theorem}

\begin{proof}
    The growth schedules can be straightly combined into a single one by appending the end of each growth schedule with the beginning of the next one, since every sub-schedule $\sigma_i$ uses only a single vertex as an initiator $u$, which is always available (i.e.,~$u$ was generated by some previous $\sigma_{j}$). Since we have $O(\log n)$ schedules and every schedule has $O(\log n)$ slots, the combined growth schedule has $O(\log^2 n)$ slots. Note that every schedule used to reverse a phase uses $O(n)$ excess edges, where $n$ is the number of vertices generated in that schedule. Since the complete schedule generates $n-1$ vertices, the excess edges activated throughout the complete schedule are $O(n)$.
\end{proof}

\subsection{Planar Graphs}

In this section, we provide an algorithm that computes a growth schedule for any target planar graph $G=(V,E)$. The algorithm first computes a $5$-coloring of $G$ and partitions the vertices into color-sets $V_i$, $1 \leq i \leq 5$. 
The color-sets are used to compute the growth schedule for $G$. The schedule contains five sub-schedules, each sub-schedule $i$ generating all vertices in color-set $V_i$. In every sub-schedule $i$, we use a modified version of the {\em \texttt{star}} schedule to generate set $V_i$.

\noindent \textbf{Pre-processing:} By using the algorithm of \cite{Williams85}, the pre-processing step computes a 5-coloring of the target planar graph $G$. This creates color-sets $V_i\subseteq V$, where $1 \leq i \leq 5$, every color-set $V_i$ containing all vertices of color $i$. W.l.o.g., we can assume that $|V_1|\geq|V_2|\geq|V_3|\geq|V_4|\geq|V_5|$. Note that every color-set $V_i$ is an independent set of $G$.

\medskip

\noindent\underline{\textbf{{\em \texttt{Planar}} algorithm}:} The algorithm picks an arbitrary vertex from $V_1$ and makes it the initiator $u_0$ of all sub-schedules. Let $V_i=\{u_1,u_2,\ldots, u_{|V_i|}\}$. For every sub-schedule $i$, $1\leq i\leq 5$, it uses the {\em \texttt{star}} schedule with $u_0$ as the initiator, to grow the vertices in $V_i$ in an arbitrary sequence, with some additional edge activations. In particular, upon generating vertex $u_x\in V_i$, for all $1\leq x\leq |V_i|$:

\begin{enumerate}
    \item Edge $vu_x$ is activated if $v\in \bigcup_{j<i} V_j$ and $u_yv\in E$, for some $u_y\in V_i\cap P_{u_x}$, both hold (recall that $P_{u_x}$ contains the descendants of $u_x$). 
    \item Edge $wu_x$ is activated if $w\in \bigcup_{j<i} V_j$ and $wu_x\in E$ both hold. 
\end{enumerate}

Once all vertices of $V_i$ have been generated, the schedule moves on to generate $V_{i+1}$. Once all vertices have been generated, the schedule deletes every edge $uv\notin E$. Note that every edge activated in the growth schedule is an excess edge with the exception of edges satisfying (2). For an edge $wu_x$ from (2) to satisfy the edge-activation distance constraint it must hold that every vertex in the birth path of $u_x$ has an edge with $w$. This holds true for the edges added in (2), due to the edges added in (1).

The edges of the {\em \texttt{star}} schedule are used to quickly generate the vertices, while the edges of (1) are used to enable the activation of the edges of (2). By proving that the {\em \texttt{star}} schedule activate $O(n)$ edges, (1) activates $O(n\log n)$ edges, and by observing that the schedule contains {\em\texttt{star}} sub-schedules that have $5\times O(\log n)$ slots in total, the next theorem follows.

\begin{lemma}\label{lem:planar-cor}
    Given a target planar graph $G=(V,E)$, the {\em \texttt{planar}} algorithm returns a growth schedule for $G$.
\end{lemma}

\begin{proof}
    Based on the description of the schedule, it is easy to see that we generate exactly $|V|$ vertices, since we break $V$ into our five sets $V_i$ and we generate each set in a different phase $i$. This is always possible no matter the graph $G$, since every set $V_i$ is an independent set.
    
    We will now prove that we also generate activate the edges of $G$. Note that this holds trivially since (2) activates exactly those edges. What remains is to argue that the edges of (2) do not violate the edge activation distance $d=2$ constraint. This constraint is satisfied by the edges activated by (1) since for every edge $wu_x\in G$, the schedule makes sure to activate every edge $uu_y$, where vertices $u_y$ are the vertices in the birth path of $u_x$. 
\end{proof}

\begin{lemma}\label{lem:planar-com}
    The {\em \texttt{planar}} algorithm has $O(\log n)$ slots and $O(n \log n)$ excess edges.
\end{lemma}

\begin{proof}
    Let $n_i$ be the size of the independent set $V_i$. Then, the sub-schedule that constructs $V_i$ requires the same number of slots as {\em \texttt{path}}, which is $\ceil{\log n_i}$ slots. Combining the five sub-schedules requires  $\sum_{i=1}^5{\log{n_i}} = \log{\prod_{i=1}^5{n_i}} < 5\log{n} = O(\log{n})$ slots.
   
    Let us consider the excess edges activated in every sub-schedule. The number of excess edges activated are the excess edges of the star schedule and the excess edges for the progeny of each vertex. The excess edges of the star schedule are $O(n)$. We also know that the progeny of each vertex $u$ includes at most $|P_u| = O(\log n)$ vertices since the length of the growth schedule is $O(\log n)$. Since we have a planar graph we know that there are at most $3n$ edges in graph $G$. For every edge $(u, v)$ in the target graph, we would need to add at most $O(\log n)$  additional excess edges. Therefore, no matter the structure of the $3n$ edges, the schedule would activate $3n O(\log n) = O(n\log n)$ excess edges.
\end{proof}

The next theorem now follows from Lemmas \ref{lem:planar-cor} and \ref{lem:planar-com}.

\begin{theorem}
For any planar graph $G$ on $n$ vertices, the {\em \texttt{planar}} algorithm computes in polynomial time a growth schedule for $G$ of $O(\log n)$ slots and $O(n\log n)$ excess edges.
\end{theorem}

\begin{definition}
    A $k$-degenerate graph $G$ is an undirected graph in which every subgraph has a vertex of degree at most $k$.
\end{definition}

\begin{corollary}
    The {\em \texttt{planar}} algorithm can be extended to compute, for any graph $G$ on $n$ vertices and in polynomial time, a growth schedule  of $O((k_1+1)\log n)$ slots, $O(k_2n\log n)$ and excess edges, where (i) $k_1=k_2$ is the degeneracy of graph $G$, or (ii) $k_1=\Delta$ is the maximum degree of graph $G$ and $k_2=|E|/n$.
\end{corollary}

\begin{proof}
    For case (i), if graph $G$ is {\em $k_1$-degenerate}, then an ordering with coloring number $k_1+1$ can be obtained by repeatedly finding a vertex $v$ with at most $x$ neighbors, removing $v$ from the graph, ordering the remaining vertices, and adding $v$ to the end of the ordering. By Lemma \ref{lem:planar-com}, the algorithm using a $k_1+1$ coloring would produce a growth schedule of $O((k_1+1)\log n)$ slots. Since graph $G$ is $k_2-degenerate$, $G$ has at most $k_2\times n$ edges and by the proof of Lemma \ref{lem:planar-com}, the algorithm would require $O(k_2n \log n)$ excess edges.
    For case (ii), we compute a $\Delta+1$ coloring using a greedy algorithm and then use the planar graph algorithm with the computed coloring as an input. By the proof of Lemma \ref{lem:planar-com}, the algorithm would produce a growth schedule of $O((\Delta+1)\log n)$ slots. 
\end{proof}

\subsection{Lower Bounds on the Excess Edges}
\label{subsec:lower-bounds}

In this section, we provide some lower bounds on the number of excess edges required to grow a graph if we fix the number of slots to $\log n$. For simplicity, we assume that $n=2^\delta$ for some integer $\delta$, but this assumption can be dropped.

We define a particular graph $G_{min}$ of size $n$, through a growth schedule $\sigma_{min}$ for it. The schedule $\sigma_{min}$ contains $\log n$ slots. In every slot $t$, the schedule generates one vertex $u^{\prime}$ for every vertex $u$ in $(G_{min})_{t-1}$ and activates $uu^{\prime}$. This completes the description of $\sigma_{min}$. Let $G$ be any graph on $n$ vertices, grown by a $\log n$-slot schedule $\sigma$. Observe that any edge activated by $\sigma_{min}$ is also activated by $\sigma$. Thus, any edges of $G_{min}$ ``not used'' by $G$ are excess edges that must be deleted by $\sigma$, for $G$ to be grown by it. The latter is captured by the following minimum edge-difference over all permutations of $V(G)$ mapped on $V(G_{min})$.

Consider the set $B$ of all possible bijections between the vertex sets of $V(G)$ and $V(G_{min})$, $b : V(G) \longmapsto V(G_{min})$. We define the edge-difference $ED_b$ of every such bijection $b\in B$ as $ED_b = |\{uv\in E(G_{min})\;|\;b(u)b(v)\notin E(G)\}|$. The minimum edge-difference over all bijections $b\in B$ is $\min\limits_{b} ED_b$. We argue that a growth schedule of $\log n$ slots for graph $G$ uses at least $\min\limits_{p} ED_p$ excess edges since the schedule has to activate every edge of $G_{min}$ and then delete at least the minimum edge-difference to get $G$.
This property leads to the following theorem, which can then be used to obtain lower bounds for specific graph families.

\begin{theorem}
    Any growth schedule $\sigma $of $\log n$ slots that grows a graph $G$ of $n$ vertices, uses at least $\min\limits_{b} ED_b$ excess edges.
\end{theorem}

\begin{proof}
    Since every schedule $\sigma$ of $\log n$ slots activates every edge $uv$ of $G_{min_i}$, $\sigma$ must delete every edge $uv\notin G$. To find the minimum number of such edges, if we consider the set $B$ of all possible bijections between the vertex sets of $V(G)$ and $V(G_{min})$, $b : V(G) \longmapsto V(G_{min})$ and we compute the minimum edge-difference over all bijections $b\in B$ as $\min\limits_{b} ED_b$, then schedule $\sigma$ has to activate every edge of $G_{min}$ and delete at least $\min\limits_{b} ED_b$ edges.
\end{proof}

\begin{corollary}\label{lower-bound-lem-line}
    Any growth schedule of $\log n$ slots that grows a path or star graph of $n$ vertices, uses $\Omega(n)$ excess edges.
\end{corollary}

\begin{proof}
    Note that for a star graph $G=(V,E)$, the maximum degree of a vertex in $G_{min}$ is $\log n$ and the star graph has a center vertex with degree $n-1$. This implies that there are $n-1-\log n$ edges of $G_{min}$ which are not in $E$. Therefore $\min\limits_{b} ED_b= (n-1-\log n)$. A similar argument works for the the schedule of a path graph.
\end{proof}

We now define a particular graph $G_{{full}}=(V,E)$ by providing a growth schedule for it. The schedule contains $\log n$ slots. In every slot $t$, the schedule generates one vertex $u^{\prime}$ for every vertex $u$ in $G_{t-1}$ and activates $uu^{\prime}$. Upon generating vertex $u^{\prime}$, it activates an edge $u^{\prime}v$ with every vertex $v$ that is within $d=2$ from $u^{\prime}$. Assume that we name the vertices $u_1,u_2,\ldots,u_n$, where vertex $u_1$ was the initiator and vertex $u_j$ was generated in slot $\ceil{\log(u_j)}$ and connected with vertex $u_{j-\ceil{\log(u_j)}}$.

\begin{lemma}
If $n$ is the number of vertices of $G_{{full}}=(V,E)$ then the number of edges of $G_{{full}}$ is $n\log n\leq |E|\leq 2n\log n$.
\end{lemma}

\begin{proof}
Let $f(x)$ be the sum of degrees when $x$ vertices have been generated. Clearly $f(2)=2$. Now consider slot $t$ and lets assume it has $x$ vertices at its end. At end of next slot we have $2x$ vertices.
Let the degrees of the vertices at end of slot $t$ be $d_1 , d_2 , ..., d_k$. Consider now that:
\begin{itemize}
    \item Child $i'$ of vertex $i$ (generated in slot $t+1$) has $1$ edge with its parent and $d_i$ edges (since an edge between it and all vertices at distance $1$ from $i$ will be activated in slot $t$. So $d_i'= d_i + 1$.
    \item Vertex $i$ has 1 edge (with its child) and $d_i$ edges (one from each new child of its neighbours in slot $t$), that is $d_i(new)=2d_i + 1$.
\end{itemize}
Therefore $f(2x) = 3 f(x) + 2x$. Notice that $2f(x) + 2x\leq f(2x) \leq 4f(x) + 4x$. Let $g(x)$ be such that $g(2)=2$ and $g(2x)=2g(x)+2x$. We claim $g(x)=x\log x$. Indeed $g(2)=2\log 2 =2$ and by induction $g(2x) = 2g(x)+2x= 2x\log x +2x = 2x\log(2x)$.
It follows that $n\log n \leq f(n) \leq 2n\log n$.
\end{proof}

We will now describe the following bipartite graph $G_{bipart}=(V,E)$ using $G_{{full}}=(V',E')$ to describe the edges of $G_{bipart}$. Both parts of the graph have $n/2$ vertices and the left part, called $A$, contains vertices $a_1,a_2,\ldots,a_{n/2}$, and the right part, called $B$, contains vertices $b_1,b_2,\ldots,b_{n/2}$, and $E'=\{a_ib_j\mid (u_i,u_j\in E)\lor (i=j)\}$. This means that if graph $G_{\text{full}}$ has $m$ edges, $G_{bipart}$ has $\Theta(m)$ edges as well.

\begin{theorem}
    Consider graph $G_{bipart}=(V',E')$ of size $n$. Any growth schedule $\sigma$ for graph $G_{bipart}$ of $\log n$ slots uses $\Omega(n \log n)$ excess edges.
\end{theorem}

\begin{proof}
    Assume that schedule $\sigma$ of $\log n$ slots, grows graph $G_{bipart}$. Since $\sigma$ has $\log n$ slots, for every vertex $u\in V'_{j-1}$ a vertex must be generated in every slot $j$ in order for the graph to have size $n$. This implies that in the last slot, $n/2$ vertices have to be generated and we remind that these vertices must be an independent set in $G_{bipart}$. For $i=\{n/2\}$, $a_i,b_i\in E'$ and both vertices $a_i,b_i$ cannot be generated together in the last slot. This implies that in the last slot, for every $i=1,2,\ldots,n/2$, we must have exactly one vertex from each pair of $a_i,b_i$. Note though that vertices $a_1,b_1$ have an edge with every vertex in $B,A$ respectively. If vertex $a_1$ or $b_1$ is generated in the last slot, only vertices from $A$ or $B$, respectively, can be generated in that same slot. Thus, we can decide that the last slot must either contain every vertex in $A$ or every vertex in $B$. 
    
    W.l.o.g., assume that in the last slot, we generate every vertex in $B$. This means that for every vertex $a_i\in A$ one vertex $b_j\in B$ must be generated. Consider an arbitrary vertex $a_i$ for which an arbitrary vertex $b_j$ is generated. In order for this to happen in the last slot, for every $a_l,b_j\in (E'\setminus a_i,b_j)$, $a_la_i$ must be active and every edge $a_la_i$ is an excess edge since set $A$ is an independent set in graph $G_{bipart}$. This means that for each vertex $b_j$ generation, any growth schedule must activate at least $deg(b_j)-1$ excess edges. By construction, graph $G_{bipart}$ has $O(n\log n)$ edges and thus, the sum of the degrees of vertices in $B$ is $O(n\log n)$. Therefore, any growth schedule has to activate $\Omega(n\log n)-n=\Omega(n\log n)$ excess edges.
\end{proof}

\section{Conclusion and Open Problems}\label{sec:Conclusion}

In this work, we considered a new model for highly dynamic networks, called growing graphs. The model, with no limitation to the edge-activation distance $d$, allows any target graph $G$ to be grown, starting from an initial singleton graph, but large values of $d$ are an impractical assumption with simple solutions and therefore we focused on cases where $d = 2$. We defined performance measures to quantify the speed (slots) and efficiency (excess edges) of the growth process, and we noticed that there is a natural trade off between the two. We proposed algorithms for general graph classes that try to balance speed and efficiency. If someone wants super efficient growth schedules (zero excess edges), it is impossible to even find a $n^{\frac{1}{3}- \varepsilon}$-approximation of the length of such a schedule, unless P~=~NP. For the special case of schedules of $\log n$ slots and $\ell=0$ excess edges, we provide a polynomial-time algorithm that can find such a schedule.

We believe that the present study, apart from opening new avenues of algorithmic research in graph-generation processes, can inspire work on more applied models of dynamic networks and network deployment, including ones in which the growth process is decentralized and exclusively controlled by the individual network processors and models whose the dynamics are constrained by geometry.

There is a number of interesting technical questions left open by the findings of this paper. It would be interesting to see whether there exists an algorithm that can decide the minimum number of edges required for any schedule to grow a graph $G$ or whether the problem is NP-hard. Note that this problem is equivalent to the cop-win completion problem; that is, $\ell$ is in this case equal to the smallest number of edges that need to be added to $G$ to make it a cop-win graph. We mostly focused on the two extremes of the $(k,\ell)$-spectrum, namely one in which $k$ is close to $\log n$ and the other is which $\ell$ close to zero. The in-between landscape remains to be explored. Finally, we gave some efficient algorithms, mostly for specific graph families, but there seems to be room for more positive results. It would also be interesting to study a combination of the growth dynamics of the present work and the edge-modification dynamics of \cite{michail2022distributed}, thus, allowing the activation of edges between vertices generated in past slots. 

\bibliographystyle{plainurl}
\bibliography{references}
\end{document}